\documentclass{article}

\usepackage{nicematrix}

\usepackage{framed}

\usepackage{lipsum}

\usepackage{tikz, pgf, pgfplots}
\usetikzlibrary
	{
		arrows, automata, chains, datavisualization.formats.functions, fadings, fit, positioning, quantikz2, quotes, shapes
	}
\pgfplotsset{compat=1.18}
\usepackage{tikzpeople}
\usepackage{fontawesome5}
\newcounter{mathseed}
\pgfmathsetseed{\arabic{mathseed}}
\pgfdeclarelayer{background}
\pgfsetlayers{background,main}
\pgfdeclaredecoration{irregular fractal line}{init}
{
	\state{init}[width=\pgfdecoratedinputsegmentremainingdistance]
	{
		\pgfpathlineto{\pgfpoint{random*\pgfdecoratedinputsegmentremainingdistance}{(random*\pgfdecorationsegmentamplitude-0.02)*\pgfdecoratedinputsegmentremainingdistance}}
		\pgfpathlineto{\pgfpoint{\pgfdecoratedinputsegmentremainingdistance}{0pt}}
	}
}
\tikzset
{
	paper/.style =
	{
		draw = MyDarkBlue!10, blur shadow, every shadow/.style = { opacity = 1, MyDarkBlue }, shading = bilinear interpolation, lower left = MyDarkBlue!10, upper left = MyDarkBlue!5, upper right = GreenTeal!75, lower right = MyDarkBlue!5, fill=none
	},
	irregular cloudy border/.style =
	{
		decoration = { irregular fractal line, amplitude = 0.2 }, decorate,
	},
	irregular spiky border/.style =
	{
		decoration = { irregular fractal line, amplitude = -0.2 }, decorate,
	},
	ragged border/.style =
	{
		decoration = {random steps, segment length = 7mm, amplitude = 2mm }, decorate
	}
}
\def\tornpaper#1{%
	\ifthenelse{\isodd{\value{mathseed}}}
	{%
		\tikz
		{
			\node[inner sep = 1em] (A) {#1};		
			\begin{pgfonlayer}{background}			
				\fill[paper]						
				\pgfextra{\pgfmathsetseed{\arabic{mathseed}}\addtocounter{mathseed}{1}}%
				{decorate[irregular cloudy border]{decorate{decorate{decorate{decorate[ragged border]{
										(A.north west) -- (A.north east)
				}}}}}}
				-- (A.south east)
				\pgfextra{\pgfmathsetseed{\arabic{mathseed}}}%
				{decorate[irregular spiky border]{decorate{decorate{decorate{decorate[ragged border]{
										-- (A.south west)
				}}}}}}
				-- (A.north west);
			\end{pgfonlayer}
		}
	}
	{%
		\tikz{
			\node[inner sep=1em] (A) {#1};  
			\begin{pgfonlayer}{background}  
				\fill[paper] 
				\pgfextra{\pgfmathsetseed{\arabic{mathseed}}\addtocounter{mathseed}{1}}%
				{decorate[irregular spiky border]{decorate{decorate{decorate{decorate[ragged border]{
										(A.north east) -- (A.north west)
				}}}}}}
				-- (A.south west)
				\pgfextra{\pgfmathsetseed{\arabic{mathseed}}}%
				{decorate[irregular cloudy border]{decorate{decorate{decorate{decorate[ragged border]{
										-- (A.south east)
				}}}}}}
				-- (A.north east);
		\end{pgfonlayer}}
	}
}

\usepackage{amsthm}
\usepackage{amssymb}
\usepackage[bb = boondox]{mathalfa}
\usepackage{dsfont}
\usepackage{newtxmath}
\usepackage{mathtools}
\usepackage{multicol}
\usepackage{braket}
\usepackage{physics}
\numberwithin{equation}{section}

\usepackage[ compat = 0.6 ]{yquant}

\usepackage[a4paper, left = 2.5cm, right = 2.5 cm, top = 2.5cm, bottom = 3.0 cm]{geometry}

\definecolor{MyLightRed}{RGB}{244, 213, 245}
\definecolor{WordRed}{RGB}{255, 0, 102}
\definecolor{RedDarkLightest}{HTML}{ff0088}
\definecolor{RedDarkLight}{HTML}{ea005f}
\definecolor{RedPurple}{HTML}{aa007f}
\definecolor{Purple}{HTML}{911146}
\definecolor{WordLightGreen}{RGB}{140, 214, 192}
\definecolor{WordGreen}{RGB}{0, 176, 80}
\definecolor{GreenLightest}{HTML}{00ffa0}
\definecolor{GreenLighter1}{HTML}{00b383}
\definecolor{GreenLighter2}{HTML}{00aa7f}
\definecolor{GreenDark}{HTML}{225522}
\definecolor{GreenTeal}{HTML}{008080}
\definecolor{WordIceBlue}{RGB}{223, 227, 229}
\definecolor{MyVeryLightBlue}{RGB}{211, 245, 247}
\definecolor{WordBlueVeryLight}{RGB}{0, 176, 240}
\definecolor{WordBlueLight}{RGB}{0, 112, 192}
\definecolor{WordBlueDark}{RGB}{46, 116, 181}
\definecolor{WordBlueDarker}{RGB}{31, 78, 121}
\definecolor{WordBlueDarker25}{RGB}{54, 96, 146}
\definecolor{WordBlueDarker50}{RGB}{36, 64, 98}
\definecolor{WordBlueDarkest}{RGB}{0, 32, 96}
\definecolor{WordBlue}{RGB}{19, 65, 99}
\definecolor{MyBlue}{RGB}{0, 64, 128}
\definecolor{MyDarkBlue}{RGB}{0, 51, 102}
\definecolor{BlueVeryDark}{HTML}{222255}
\definecolor{MagentaVeryLight}{RGB}{178, 162, 201}
\definecolor{MagentaLighter}{RGB}{161, 106, 221}
\definecolor{MagentaLight}{RGB}{128, 100, 162}
\definecolor{MagentaDark}{RGB}{106, 65, 152}
\definecolor{MagentaVeryDark}{RGB}{97, 75, 128}
\definecolor{WordAquaLighter80}{RGB}{218, 238, 243}
\definecolor{WordAquaLighter60}{RGB}{183, 222, 232}
\definecolor{WordAquaLighter40}{RGB}{146, 205, 220}
\definecolor{WordAquaDarker25}{RGB}{49, 134, 155}
\definecolor{WordAquaDarker50}{RGB}{33, 89, 103}
\definecolor{WordVeryLightTeal}{RGB}{223, 236, 235}
\definecolor{WordLightTeal}{RGB}{160, 199, 197}
\definecolor{WordDarkTealLighter80}{RGB}{207, 223, 234}
\definecolor{WordDarkTeal}{RGB}{72, 123, 119}
\definecolor{WordDarkerTeal}{RGB}{48, 82, 80}
\definecolor{WordTurquoiseLighter80}{RGB}{209, 238, 249}
\definecolor{Brown}{HTML}{666633}

\usepackage{xcolor}
\usepackage{varwidth}
\usepackage[skins, breakable, listings, theorems] {tcolorbox}
\usepackage{graphicx}
\usepackage{xurl}                                                                                  %
\usepackage{hyperref}
\hypersetup
{
	colorlinks,
	linkcolor = {WordRed!75!black},
	citecolor = {WordBlueDarker25},
	urlcolor = {WordAquaDarker50}
}
\usepackage{colortbl}
\usepackage{float}
\usepackage{cellspace}
\usepackage{makecell}
\usepackage{multirow}
\usepackage[linesnumbered, tworuled, vlined]{algorithm2e}
\usepackage{appendix}
\usepackage{enumitem}
\usepackage{xintexpr}
\newcommand{\orcidicon}[1]{\href{https://orcid.org/#1}{\includegraphics[height=\fontcharht\font`\B]{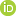}}}

\newtheorem{theorem}{Theorem}[section]

\newtheorem{corollary}[theorem]{Corollary}
\newtheorem{example}{Example}[section]

\title
	{
		Quantum Tapsilou - a quantum game inspired from the traditional Greek coin tossing game tapsilou
	}

\author
	{
		Kalliopi Kastampolidou\orcidicon{0000-0003-3607-9569}
		and
		Theodore Andronikos\orcidicon{0000-0002-3741-1271}
		\\
		Department of Informatics, Ionian University, \\
		7 Tsirigoti Square, 49100 Corfu, Greece; \\
		\{c17kast, andronikos\}@ionio.gr
	}

\begin{document}

\maketitle

\begin{abstract}
	This paper introduces a new quantum game called Quantum Tapsilou that is inspired by the classical traditional Greek coin tossing game tapsilou. The new quantum game, despite its increased complexity and scope, retains the most important characteristic of the traditional game. In the classical game, both players have $\frac { 1 } { 4 }$ probability to win. The quantum version retains this characteristic feature, that is both players have the same probability to win, only now this probability varies considerably and depends on previous moves and choices. The two most important novelties of Quantum Tapsilou can be attributed to its implementation of entanglement via the use of rotation gates instead of Hadamard gates, which generates Bell-like states with unequal probability amplitudes, and the integral use of groups. In Quantum Tapsilou both players agree on a specific cyclic rotation group of order $n$, for some sufficiently large $n$. The game is based on the chosen group, in the sense that both players will draw their moves from its elements. More specifically, both players will pick rotations from this group to realize their actions using the corresponding $R_{ y }$ rotation gates. In the Quantum Tapsilou game, it is equally probable for both players to win. This fact is in accordance with a previous result in the literature showing that quantum games where both players choose their actions from the same group, exhibit perfect symmetry by providing each player with the possibility to pick the move that counteracts the other player's action.
	\\
\textbf{Keywords:}: Classical games, quantum games, quantum coin tossing games, quantum strategies.
\end{abstract}
\section{Introduction} \label{sec: Introduction}

Undoubtedly, game theory has established itself by now as a multifaceted framework that can readily and easily provide a multitude of tools that facilitate the analysis of complex situations, especially those concerning conflict, competition or cooperation. The actors or agents involved in these situations are called players. The fundamental assumption of game theory is that players are always rational, and strive to maximize their profit or minimize their loss. Each player possesses their own set of objectives, and their actions can disrupt the plans of other players. At the conclusion of the game, each player assesses their gains or losses by employing a payoff function. The premise of rationality plays a crucial role in guaranteeing that each agent is motivated to maximize their own payoff. Game theory is particularly suited for tackling serious and difficult problems from a plethora of different fields, ranging from economics, political and social sciences, to computer science, biology and psychology. It is instructive to mention as an example of social and legal significance the possibility to understand and detect criminal activities via the use of evolutionary game theory, as proposed in \cite{Perc2013}. The reader interested in the study of similar critical social issues may consult \cite{Helbing2014} for more references. Evolutionary dynamics offers fresh and unconventional insights to classical games such as the famous Prisoners' Dilemma (see \cite{Perc2008} and \cite{Szolnoki2008} for more details). Additional references on the subject of evolutionary and coevolutionary games can be found in \cite{Perc2013a} and \cite{Perc2010}. In this work we shall be employing only a few elementary concepts from game theory that can be found in standard textbooks such as \cite{Dixit2015}, \cite{Myerson1997} or \cite{Maschler2020}.

The realm of quantum computation holds the promise of ushering us in a new era of computational capabilities. Despite its origins dating back to the early 1980s, it took some time for this field to fully develop and reach its current state. Presently, there are commercial quantum computers in existence, albeit with a rather modest number of qubits. However, it is anticipated that this situation will undergo significant transformation in the years ahead. If these predictions prove accurate, there will be a substantial increase in the number of available qubits. Consequently, this quantum leap will lead to a remarkable enhancement of our computational power, enabling us to address complex decision-making and optimization challenges effectively.

In this broader context, it appears inevitable that quantum games would emerge as a natural consequence. The emerging field of quantum game theory explores the integration of quantum principles and concepts into classical games, such as coin flipping, the Prisoner's Dilemma, and various others.

\subsection{Related work} \label{subsec: Related Work}

It would be historically accurate to say that the inception of the field of quantum games occurred in 1999 when Meyer introduced the PQ penny flip game in his influential paper \cite{Meyer1999}. This game serves as the quantum counterpart to the classical penny flip game. One of the notable features of Meyer's formulation was the incorporation of two well-known fictional characters, Picard and Q, from the renowned television series Star Trek. In this game, Picard embodies the classical player, while Q embodies the quantum player. This signifies that Picard follows the rules of the classical penny flip game, while Q can utilize arbitrary unitary operators, a hallmark of the quantum realm. Meyer demonstrated that in this game, Q can consistently emerge victorious by employing the venerable Hadamard transform. Another influential work, also from the same year, was the one by Eisert et al. in \cite{Eisert1999}. There they demonstrated for the first time the so called Eisert-Wilkens-Lewenstein protocol. This protocol has gained wide acceptance and is now extensively used in the literature. Using their protocol, they defined a variant of the famous Prisoners' Dilemma and demonstrated that there exists a quantum strategy that is better than any classical strategy. Subsequently, many researchers followed that line of research, obtaining interesting results. For more recent developments we refer to \cite{Giannakis2019} and \cite{Rycerz2020}, where the Eisert–Wilkens–Lewenstein protocol was generalized.

Subsequently, many researchers, inspired by Meyer's concepts, extended this game pattern to $n$-dimensional quantum systems. This avenue of exploration was pursued in \cite{Wang2000, Ren2007} and \cite{Salimi2009}. Later, a comprehensive mapping of every finite variation of the PQ penny flip game to easily constructed finite automata was provided in \cite{Andronikos2018}. Recently, this same game underwent extensive analysis from the perspective of group theory, with its relation to dihedral groups being established in \cite{Andronikos2021}. A general analysis of quantum games under the assumption that both players draw their move from groups and not just unstructured sets, coming to some interesting conclusions, can be found in \cite{Andronikos2022a}. While in many works, the quantum player appears to possess a distinct advantage over the classical player, the situation is not as straightforward, as demonstrated in \cite{Anand2015}. In that study, the authors cleverly altered the rules of the PQ penny flip game, allowing the classical player to secure a victory.

Additionally, a broader and intriguing problem concerning quantum gambling based on Nash equilibrium has been explored recently in \cite{Zhang2017}. The connection between infinitely repeated games and finite automata was also studied by \cite{Neyman1985}, \cite{Rubinstein1986}, \cite{Abreu1988}, and \cite{Marks1990}. Later, Meyer used quantum lattice gas automata to study Parrondo games in \cite{Meyer2002}. Employing probabilistic automata to tackle the Prisoners' Dilemma game was undertaken by \cite{Bertelle2002}, whereas the authors in \cite{Suwais2014} considered the advantages of different automata variants in game theory. In a similar vein, quantum coin flipping has served as a crucial component in numerous quantum cryptographic protocols, where familiar figures like Alice and Bob take on the roles of remote parties striving to reach a consensus on a random bit. For a comprehensive presentation, the seminal work in \cite{Bennett2014} is invaluable, and for some state-of-the-art cryprographic protocols we refer to \cite{Ampatzis2021, Ampatzis2022, Ampatzis2023, Andronikos2023, Andronikos2023a}). Some of these ideas have been extended in \cite{Aharon2010} to quantum dice rolling in scenarios involving multiple outcomes and more than two parties. Most of the papers mentioned above, rely in an integral way on the phenomenon of quantum entanglement, a trend that is also followed by the current work.

Another influential work, also from 1999, was the one by Eisert et al. in \cite{Eisert1999}. There they demonstrated for the first time the so called Eisert-Wilkens-Lewenstein protocol. This protocol has gained wide acceptance and is now extensively used in the literature \cite{AlonsoSanz2019}. Using their protocol, they defined a variant of the famous Prisoners' Dilemma and demonstrated that there exists a quantum strategy that is better than any classical strategy. Subsequently, many researchers followed that line of research, obtaining interesting results. For some recent developments one may also consult \cite{Rycerz2020}, where the Eisert–Wilkens–Lewenstein protocol was generalized. Let us also point out that the inclination to employ unconventional approaches for improving outcomes in classical game s extends beyond the realm of quantum mechanics. Even in the domain of well-known classical games, such as the Prisoners' Dilemma, efforts have been made to interpret them in terms of biological concepts (for a comprehensive overview, refer to \cite{Kastampolidou2020}).

Let us also point out that the inclination to employ unconventional approaches for improving outcomes in classical games extends beyond the realm of quantum mechanics. Even in the domain of well-known classical games, such as the Prisoners' Dilemma, efforts have been made to interpret them in terms of biological concepts (for a comprehensive overview, refer to \cite{Kastampolidou2020}). In fact, many classical games can be reinterpreted within the context of biological and bio-inspired processes, as evidenced by numerous references (for more sources, see \cite{Kastampolidou2020, Theocharopoulou2019, Kastampolidou2020a}).

\textbf{Contribution}. This paper introduces Quantum Tapsilou, which is a novel quantum game inspired by a traditional Greek coin tossing game called tapsilou. While the classical game has a simple and fair probability distribution, with both players having a $\frac { 1 } { 4 }$ probability of winning, Quantum Tapsilou exhibits an additional level of complexity using quantum mechanics. In Quantum Tapsilou, both players still have equal chances of winning, but these probabilities are now influenced by their previous choices and the use of quantum principles. There are two key innovations in Quantum Tapsilou:

\begin{enumerate}
	\item	\textbf{Entanglement through Rotation Gates}. Instead of using Hadamard gates, which are commonly associated with creating entangled states with equal probability amplitudes, Quantum Tapsilou employs rotation gates, specifically $R_y$ rotation gates. These gates generate Bell-like states with unequal probability amplitudes, which adds a layer of complexity to the game dynamics.
	\item	\textbf{Integral Use of Groups}. In Quantum Tapsilou, both players agree on a specific cyclic rotation group of order $n$, where $n$ should be sufficiently large integer in order to provide the players with additional choices and increase the suspense of the game. This group forms the basis of the game, and both players select rotations from this group to perform their actions using the corresponding $R_y$ rotation gates. The choice of this group and the rotations within it significantly impacts the game's outcome, introducing a strategic element based on group theory.
\end{enumerate}

By incorporating these quantum principles and group theory, Quantum Tapsilou retains the characteristic aspects of the classical tapsilou game while adding complexity and variability to the probabilities of winning. Players must consider their choices of rotations within the agreed-upon group, as well as the entangled states generated by the $R_y$ rotation gates, to maximize their chances of winning in this intriguing quantum adaptation of a traditional Greek game. This is another example of the critical role of groups in ensuring a symmetric correspondence of actions in quantum games in the sense that for each action of each player, the other player can come up with a move to counteract it, a fact that is in accordance with a previous result in \cite{Andronikos2022a}.

\subsection*{Organization} \label{subsec: Organization}

This paper is organized as follows: Section \ref{sec: Introduction} gives a concise introduction to the field of quantum games, along with some of the most relevant references. Section \ref{sec: The Traditional Game Tapsilou} presents the traditional Greek game tapsilou, while section \ref{sec: Two Coins Seem To Be The Proper Choice} explains why the original game is meaningful and well-conceived.
Sections \ref{sec: The Intuition Behind the Quantum Tapsilou} and \ref{sec: The Formal Presentation of the QTG} give an easygoing account of the intuition behind the quantization of the classical game, and the formal presentation of the new quantum game, respectively. Section \ref{sec: The Mathematical Analysis of the QTG} contains an extensive and rigorous analysis of the main properties of the proposed game. Finally, Section \ref{sec: Discussion and Conclusions} summarizes the main conclusions of this work and contains some possible research directions for future work.

\section{The traditional game ``tapsilou''} \label{sec: The Traditional Game Tapsilou}

The Greek gambling game ``tapsilou'' is a coin tossing game that originates from the Greek island of Corfu. The name ``tapsilou'' is a local idiom meaning ``up high.'' This is a traditional game played in New Year’s Eve and New Year’s Day in Mantouki, Corfu. It is also known as ``crosses and marks'' (``stavri or markoi'') and goes back in the history of the island. It was extremely popular in the 60s, 70s and 80s among the young people and it was tradition to be played around the New Year celebration. Although it is considered a game, for many of the participants meant their economic ruin \cite{drepaniwebsite}.

Tapsilou needs two coins. They used mostly Greek currency (drachmas), the so called 5-drachmas (``talara or talira''), like those depicted in Figure \ref{fig:Drachmae}. It was the biggest in diameter coin of the days. The coin depicts a figure in one side, the \emph{heads} side, and a number on the other, the \emph{tails} side. Thus, the alternative name ``crosses and marks.'' As a terrain, the participants typically chose a sandlot. Two coins are placed in one of the player’s two-or-three fingers, as shown in Figure \ref{fig:Coin Tossing}. This player is the \emph{tosser} and from now on will also be referred to as \emph{player 1}. The initial position of the coins on the fingers is one coin with the face up and one with the cross. The other players, usually more than one, are the \emph{gamblers}, and, henceforth, will also be collectively referred to as \emph{player 2}. The gamblers draw a circle around the tosser, who throws the coins high to the sky, thus the origin of the name ``tapsilou.''  After the flip, the coins end up on the ground. If the coins end up with two faces, the tosser wins. If the coins end up with crosses, the gamblers win, and in case the coins end up with one face and one cross, the tossing is repeated until there is a winner.

\begin{tcolorbox}
	[
		grow to left by = 0.00 cm,
		grow to right by = 0.00 cm,
		colback = WordTurquoiseLighter80!12,	
		enhanced jigsaw,						
		sharp corners,
		toprule = 1.0 pt,
		bottomrule = 1.0 pt,
		leftrule = 0.1 pt,
		rightrule = 0.1 pt,
		sharp corners,
		center title,
		fonttitle = \bfseries
	]
	\begin{figure}[H]
		\centering
		\includegraphics[scale = 0.50, trim = {0cm 0cm 0cm 0cm}, clip]{"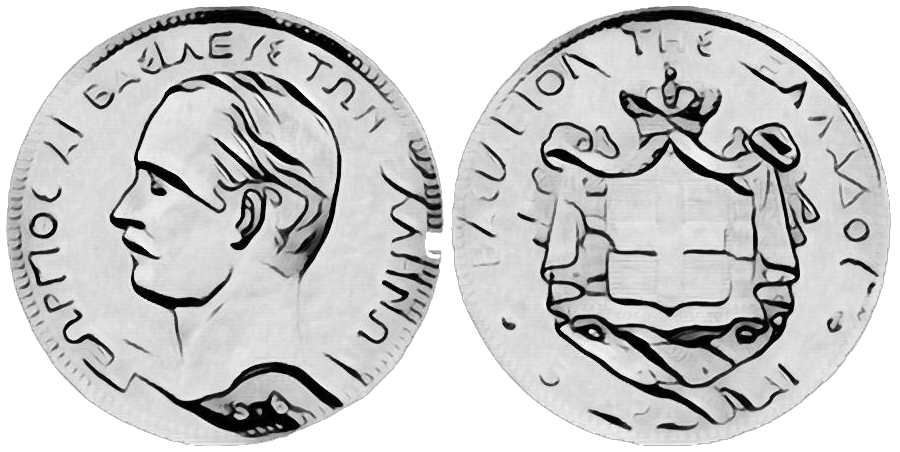"}
		\caption{A typical example of coins from the 60s \cite{collectorwebsite}.}
		\label{fig:Drachmae}
	\end{figure}
\end{tcolorbox}

\begin{tcolorbox}
	[
		grow to left by = 0.00 cm,
		grow to right by = 0.00 cm,
		colback = WordTurquoiseLighter80!12,	
		enhanced jigsaw,						
		sharp corners,
		toprule = 1.0 pt,
		bottomrule = 1.0 pt,
		leftrule = 0.1 pt,
		rightrule = 0.1 pt,
		sharp corners,
		center title,
		fonttitle = \bfseries
	]
	\begin{figure}[H]
		\centering
		\includegraphics[scale = 0.75, trim = {0cm 0cm 0cm 0cm}, clip]{"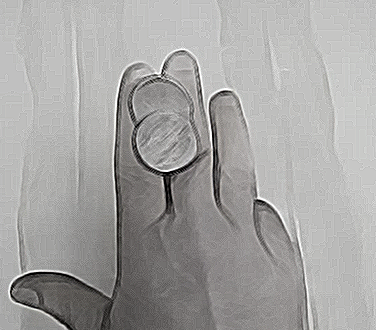"}
		\caption{The initial placement of the two coins, just before the tossing \cite{corfuwebsite}.}
		\label{fig:Coin Tossing}
	\end{figure}
\end{tcolorbox}

Although tapsilou requires at least 2 players, the tosser and at least one gambler, there is no limit as to what the maximum number of players is. It would not be unusual to have fifteen or more players, as long as there is an amount of money, large enough to bet. Before the tossing, the gamblers must announce their bet. If the two coins come up faces up, the tosser wins and takes all the bet money, whereas if they both come up crosses, the tosser pays each gambler twice the agreed bet. If the result is one coin face up and one coin cross, then the gamblers choose, individually, whether they will keep the same bet or raise it, and the tosser throws the coins again. This process is repeated, until the tosser or the gamblers win. The tosser loses when he has no more money left. If the tossers budget is running low, he may decide to take no more bets, or up to a certain amount, so as to not end up in debt. When the tosser loses all his money, another player takes his place.


Tapsilou is a forbidden game, so in many cases there may be lookout men. There are two common cheating strategies in this game. The most tricky one is to crack the coin just a little bit, so as to always drop from one side. This, of course, makes the coins biased, but the crack has to be unnoticeable. The other cheating strategy is during the tossing to employ a technique where the two coins seem to flip in the air before they end up in the ground, but in fact only one flips, while the other does not flip at all. This is just an illusion that only the best trained cheaters can achieve. If a gambler recognizes a certain cheat strategy, he has the right to call it (using Greek idioms such as``den paei'' or ``ta kovo'') and the bet is canceled.

\section{Two coins seem to be the proper choice} \label{sec: Two Coins Seem To Be The Proper Choice}

In this section, we give a precise quantitative description of the game, and examine its possible generalizations.

Assuming that no cheating takes place, we see that the classical version of tapsilou is characterized by the following properties, where we use $H$ and $T$ to denote heads and tails, respectively.

\begin{enumerate} [ left = 0.50 cm, labelsep = 1.00 cm ]
	\renewcommand\labelenumi{(\textbf{C}$_\theenumi$)}
	\item	The number of players is $2$. Player $1$ is the tosser, and player $2$ is the gambler. Despite the fact that in real life many persons may play the gambler simultaneously, they can all be identified as one player in a theoretical analysis.
	\item	The number of coins is $2$.
	\item	Viewing the two coins as a system, with no particular ordering between them, the initial state of the system can taken to be $(H, T)$ or, equivalently, $(T, H)$. So, without loss of generality, we may assume from now on that the initial state is $(H, T)$.
	\item	If the final state is $(H, H)$, then player $1$ wins. If the final state is $(T, T)$, then player $2$ wins. In any other case, the game goes on to the next round.
	\item	The probability that player $1$ wins is equal to the probability that player $2$ wins, namely $\frac { 1 } { 4 }$.
	\item	The probability that neither player wins is $\frac { 1 } { 2 }$.
\end{enumerate}

A question that comes to mind is whether it would be possible to generalize this game by tossing more than two coins. Under the plausible assumption that in such a generalization player $1$ wins if all coins end up heads and player $2$ wins if all coins end up tails, it is easy to calculate the probability for each player to win. If $n$ coins are used, the condition for player $1$ to win is that all $n$ coins end up heads, an event that may happen with probability $\frac { 1 } { 2^{ n } }$. Symmetrically, for player $2$ the probability that all $n$ coins come up tails is also $\frac { 1 } { 2^{ n } }$. The next Table \ref{tbl:Probabilities of Classical Generalized Tapsilou} gives the numerical values of these probabilities when the number of coins is $3$, $4$ and $5$. A clear pattern that emerges. If tapsilou were to use more than two coins, the game would not be very interesting because the probability of neither player winning and just spending most of the gaming tossing coins would be far greater, than having a winner. Therefore, using exactly two coins seems to be the best choice in order to have a game with suspense and symmetry.

\begin{tcolorbox}
	[
		grow to left by = 0.00 cm,
		grow to right by = 0.00 cm,
		colback = WordIceBlue!25,	%
		enhanced jigsaw,			
		sharp corners,
		boxrule = 0.1 pt,
		toprule = 0.1 pt,
		bottomrule = 0.1 pt
	]
	\begin{table}[H]
		\renewcommand{\arraystretch}{1.60}
		\caption{This table gives the probabilities of players $1$ and $2$ to win in a possible generalization of the game tapsilou.}
		\label{tbl:Probabilities of Classical Generalized Tapsilou}
		\centering
		\begin{tabular}
			{
				>{\centering\arraybackslash} m{2.00 cm} !{\vrule width 0.5 pt}
				>{\centering\arraybackslash} m{3.00 cm} !{\vrule width 0.5 pt}
				>{\centering\arraybackslash} m{3.00 cm} !{\vrule width 0.5 pt}
				>{\centering\arraybackslash} m{3.00 cm}
			}
			\Xhline{4\arrayrulewidth}
			\multicolumn{4}{c}
			{ \cellcolor[HTML]{000000} { \color[HTML]{FFFFFF} \textbf{Coins \& Probability} } }
			\\
			\Xhline{\arrayrulewidth}
			{ \cellcolor[HTML]{000000} { \color[HTML]{FFFFFF} \textbf{\# of Coins} } }
			&
			$P( \text{ player 1 wins } )$
			&
			$P( \text{ player 2 wins } )$
			&
			$P( \text{ neither wins } )$
			\\
			\Xhline{3\arrayrulewidth}
			{ \cellcolor[HTML]{000000} { \color[HTML]{FFFFFF} $3$ } }
			&
			\xintDigits*:=3;
			\xintdeffloatfunc	Events ( m ) := 2^{ m };
			\xintdeffloatfunc	Probability ( m ) := 1 / Events ( m );
			$\xintfloateval{ Probability ( 3 ) }$
			&
			\xintDigits*:=3;
			\xintdeffloatfunc	Events ( m ) := 2^{ m };
			\xintdeffloatfunc	Probability ( m ) := 1 / Events ( m );
			$\xintfloateval{ Probability ( 3 ) }$
			&
			\xintDigits*:=3;
			\xintdeffloatfunc	Events ( m ) := 2^m;
			\xintdeffloatfunc	Probability ( m ) := 1 - ( 2 / Events ( m ) );
			\xintfloateval{ Probability ( 3 ) }
			\\
			\Xhline{\arrayrulewidth}
			{ \cellcolor[HTML]{000000} { \color[HTML]{FFFFFF} $4$ } }
			&
			\xintDigits*:=3;
			\xintdeffloatfunc	Events ( m ) := 2^{ m };
			\xintdeffloatfunc	Probability ( m ) := 1 / Events ( m );
			$\xintfloateval{ Probability ( 4 ) }$
			&
			\xintDigits*:=3;
			\xintdeffloatfunc	Events ( m ) := 2^{ m };
			\xintdeffloatfunc	Probability ( m ) := 1 / Events ( m );
			$\xintfloateval{ Probability ( 4 ) }$
			&
			\xintDigits*:=3;
			\xintdeffloatfunc	Events ( m ) := 2^m;
			\xintdeffloatfunc	Probability ( m ) := 1 - ( 2 / Events ( m ) );
			\xintfloateval{ Probability ( 4 ) }
			\\
			\Xhline{\arrayrulewidth}
			{ \cellcolor[HTML]{000000} { \color[HTML]{FFFFFF} $5$ } }
			&
			\xintDigits*:=3;
			\xintdeffloatfunc	Events ( m ) := 2^{ m };
			\xintdeffloatfunc	Probability ( m ) := 1 / Events ( m );
			$\xintfloateval{ Probability ( 5 ) }$
			&
			\xintDigits*:=3;
			\xintdeffloatfunc	Events ( m ) := 2^{ m };
			\xintdeffloatfunc	Probability ( m ) := 1 / Events ( m );
			$\xintfloateval{ Probability ( 5 ) }$
			&
			\xintDigits*:=3;
			\xintdeffloatfunc	Events ( m ) := 2^m;
			\xintdeffloatfunc	Probability ( m ) := 1 - ( 2 / Events ( m ) );
			\xintfloateval{ Probability ( 5 ) }
			\\
			\Xhline{4\arrayrulewidth}
		\end{tabular}
		\renewcommand{\arraystretch}{1.0}
	\end{table}
\end{tcolorbox}

In view of our previous analysis, it would seem that the classical world cannot offer a meaningful generalization of tapsilou. If one hopes to achieve an interesting variation of the classical game, turning to the quantum realm seems to be the most promising option.

\section{The intuition behind the quantum tapsilou} \label{sec: The Intuition Behind the Quantum Tapsilou}

In this section we explain the rationale behind the proposed Quantum Tapsilou game (QTG from now), which, we believe, that captures the spirit of the classical tapsilou game. In the QTG, kets $\ket{ 0 }$ and $\ket{ 1 }$ correspond to the classical heads and tails states, respectively. Consequently, the initial state of the QTG is $\ket{ 0 } \ket{ 1 }$, and player $1$ (player $2$) wins if the final state of the system (after measurement) is $\ket{ 0 } \ket{ 0 }$ ($\ket{ 1 } \ket{ 1 }$). The following Table \ref{tbl: Classical vs Quantum Tapsilou Characteristics} summarizes these facts.

\begin{tcolorbox}
	[
		grow to left by = 0.00 cm,
		grow to right by = 0.00 cm,
		colback = WordIceBlue!25,	%
		enhanced jigsaw,			
		sharp corners,
		boxrule = 0.1 pt,
		toprule = 0.1 pt,
		bottomrule = 0.1 pt
	]
	\begin{table}[H]
		\renewcommand{\arraystretch}{1.60}
		\caption{This table compares the characteristics of the classical tapsilou against the corresponding characteristics of the quantum tapsilou.}
		\label{tbl: Classical vs Quantum Tapsilou Characteristics}
		\centering
		\begin{tabular}
			{
				>{\centering\arraybackslash} m{1.75 cm} !{\vrule width 0.5 pt}
				>{\centering\arraybackslash} m{1.75 cm} !{\vrule width 0.5 pt}
				>{\centering\arraybackslash} m{1.75 cm} !{\vrule width 0.5 pt}
				>{\centering\arraybackslash} m{2.25 cm} !{\vrule width 0.5 pt}
				>{\centering\arraybackslash} m{2.25 cm}
			}
			\Xhline{4\arrayrulewidth}
			\multicolumn{5}{c}
			{ \cellcolor[HTML]{000000} { \color[HTML]{FFFFFF} \textbf{Classical vs quantum tapsilou} } }
			\\
			\Xhline{\arrayrulewidth}
			{ \cellcolor[HTML]{000000} { \color[HTML]{FFFFFF} } }
			&
			State
			&
			Initial State
			&
			Winning State for Player $1$
			&
			Winning State for Player $2$
			\\
			\Xhline{3\arrayrulewidth}
			{ \cellcolor[HTML]{000000} { \color[HTML]{FFFFFF} \textbf{Classical} } }
			&
			$H$ (Heads) $T$ (Tails)
			&
			$(H, T)$
			&
			$(H, H)$
			&
			$(T, T)$
			\\
			\Xhline{\arrayrulewidth}
			{ \cellcolor[HTML]{000000} { \color[HTML]{FFFFFF} \textbf{Quantum} } }
			&
			$\ket{ 0 }$ \qquad \qquad $\ket{ 1 }$
			&
			$\ket{ 0 } \ket{ 1 }$
			&
			$\ket{ 0 } \ket{ 0 }$
			&
			$\ket{ 1 } \ket{ 1 }$
			\\
			\Xhline{4\arrayrulewidth}
		\end{tabular}
		\renewcommand{\arraystretch}{1.0}
	\end{table}
\end{tcolorbox}

Let us first assume that player $1$ (the tosser) can employ quantum actions to his advantage. In order for him to win, both coins must end up heads. This is easily achievable in the quantum realm using entanglement. Thus, entanglement is the key for player $1$ to ensure that both coins end up in the same state. The typical quantum circuit used to generate entangled pairs, which is based on the Hadamard transform, is shown in Figure \ref{fig: Phi-QC}. The quantum circuit of Figure \ref{fig: Phi-QC}, and all forthcoming quantum circuits, are designed using the IBM Quantum Composer \cite{IBMQuantumComposer2023}. The same tool is used for the visualization of the state vector of the quantum system. This particular circuit will drive the system of the two coins in the $\ket{ \Phi^{ - } }$ state, which is characterized by the state vector depicted in Figure \ref{fig: Phi-SV}. However, if the coins are measured in this state, the tosser does not have any real advantage, as both players have equal probability $\frac { 1 } { 2 }$ to win.

\begin{tcolorbox}
	[
		grow to left by = 1.50 cm,
		grow to right by = 1.50 cm,
		colback = white,	
		enhanced jigsaw,	
		sharp corners,
		toprule = 1.0 pt,
		bottomrule = 1.0 pt,
		leftrule = 0.1 pt,
		rightrule = 0.1 pt,
		sharp corners,
		center title,
		fonttitle = \bfseries
	]
	\begin{figure}[H]
		\begin{minipage}[t]{0.45 \textwidth}
			\centering
			\includegraphics[scale = 0.45, trim = {0cm 0cm 0cm 0cm}, clip]{"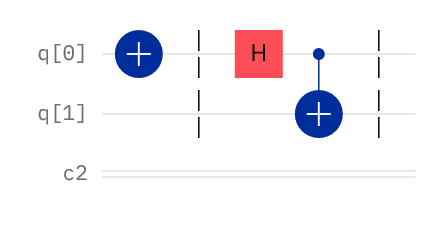"}
			\caption{This figure shows a typical quantum circuit that generates $\ket{ \Phi^{ - } }$ pairs, given that the initial state is $\ket{ 0 } \ket{ 1 }$.}
			\label{fig: Phi-QC}
		\end{minipage}
		\hfill
		\begin{minipage}[t]{0.5 \textwidth}
			\centering
			\includegraphics[scale = 0.21, trim = {0cm 6cm 16cm 0cm}, clip]{"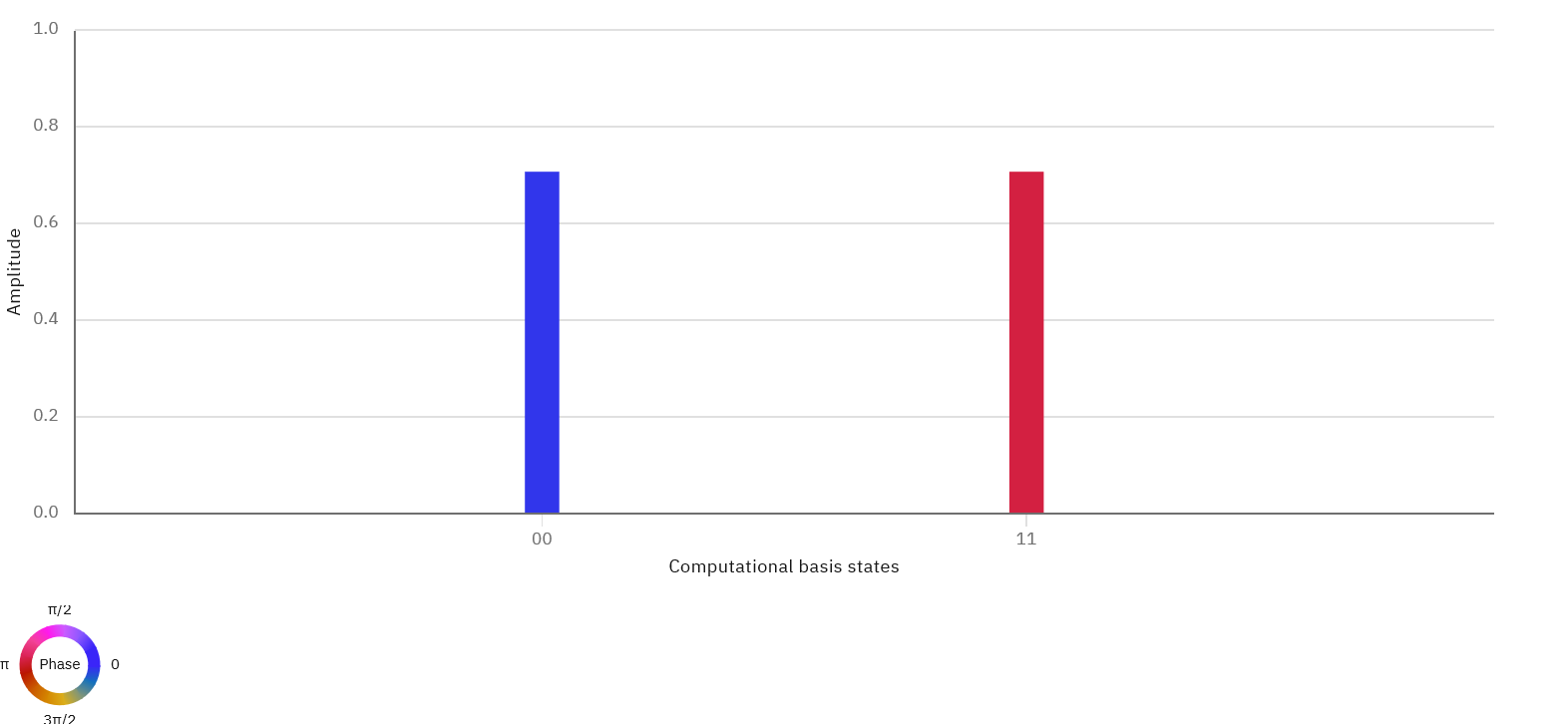"}
			\caption{This figure depicts the state vector of the system of the two coins entangled in the $\ket{ \Phi^{ - } }$ state. }
			\label{fig: Phi-SV}
		\end{minipage}
	\end{figure}
\end{tcolorbox}

Additionally, player $1$ can drive the system to the $\ket{ \Psi^{ - } }$, via the circuit of Figure \ref{fig: Psi-QC}. As shown in Figure \ref{fig: Psi-SV}, if the coins are measured in this state, it is certain that nobody wins, and the game keeps going on ad infinitum. This is another option, albeit less likely, that the tosser may pursue as a bluff.

\begin{tcolorbox}
	[
		grow to left by = 1.50 cm,
		grow to right by = 1.50 cm,
		colback = white,	
		enhanced jigsaw,	
		sharp corners,
		toprule = 1.0 pt,
		bottomrule = 1.0 pt,
		leftrule = 0.1 pt,
		rightrule = 0.1 pt,
		sharp corners,
		center title,
		fonttitle = \bfseries
	]
	\begin{figure}[H]
		\begin{minipage}[t]{0.45 \textwidth}
			\centering
			\includegraphics[scale = 0.39, trim = {0cm 0cm 0cm 0cm}, clip]{"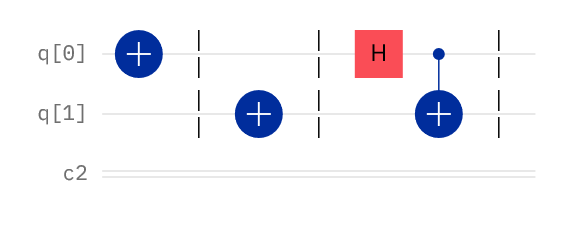"}
			\caption{Assuming that the initial state is $\ket{ 0 } \ket{ 1 }$, the above figure depicts a typical quantum circuit that generates $\ket{ \Psi^{ - } }$ pairs.}
			\label{fig: Psi-QC}
		\end{minipage}
		\hfill
		\begin{minipage}[t]{0.5 \textwidth}
			\centering
			\includegraphics[scale = 0.21, trim = {0cm 6cm 16cm 0cm}, clip]{"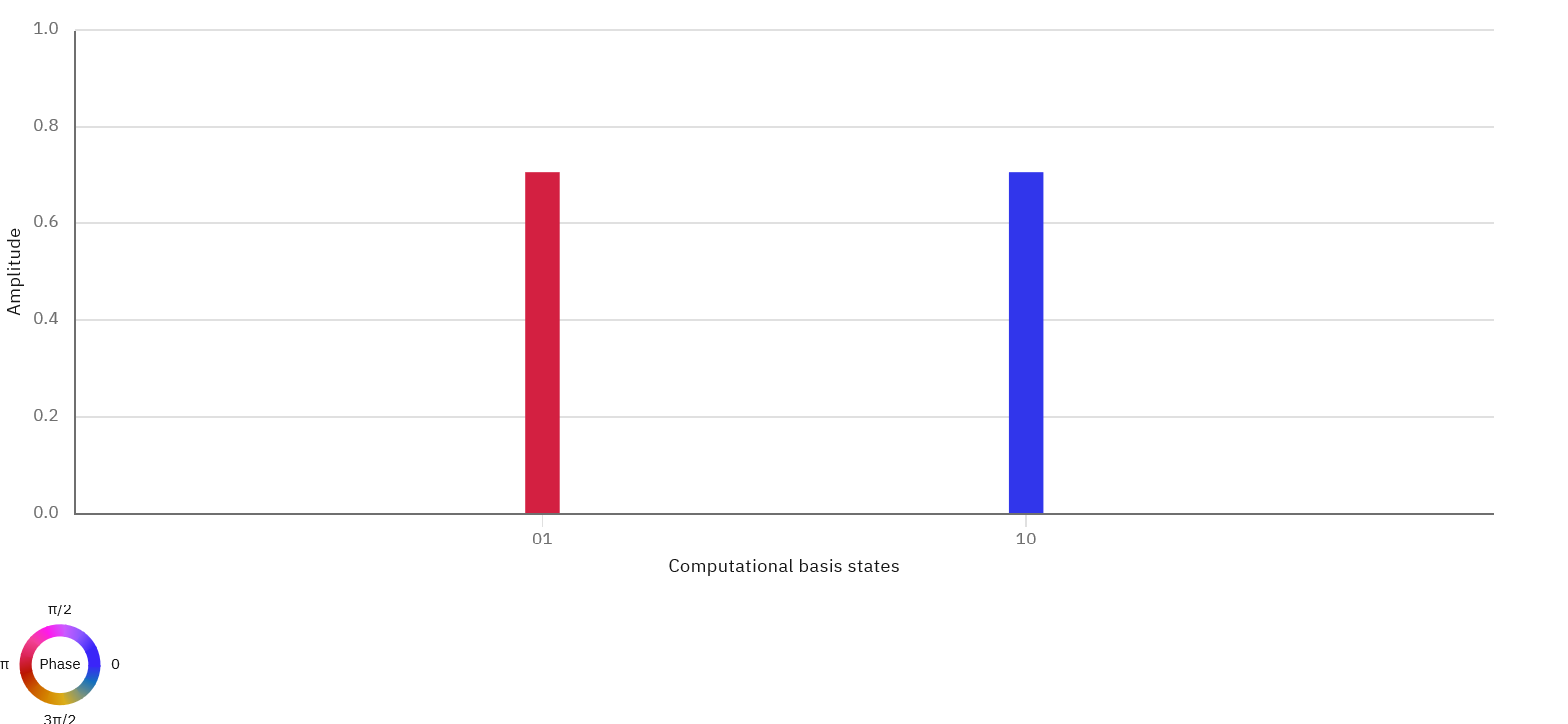"}
			\caption{The above figure visualizes the state vector of the two coin system in the $\ket{ \Psi^{ - } }$ state. }
			\label{fig: Psi-SV}
		\end{minipage}
	\end{figure}
\end{tcolorbox}

Therefore, from the perspective of player $1$, it becomes evident that he must entangle the two coins in a more subtle way. Logic suggests that player $1$ should devise a kind of entanglement that satisfies the following requirements.

\begin{itemize}
	\item	Both coins end up in the \emph{same} state with certainty, i.e., with probability $1$.
	\item	While it may still be possible that both coins may come up tails, the probability that both coins may end up heads must be significantly higher. In other words, the odds should favor player $1$ decisively.
\end{itemize}

One way that player $1$ may achieve such an effect is to use rotation gates in place of the Hadamard gate. Towards this end, the simplest choice would be to use the quantum gate $R_{ y } ( \theta )$ with an  appropriate parameter $\theta \in \mathbb{R}$. $R_{ y } ( \theta )$, which is defined as

\begin{align} \label{eq: Ry Rotation Operator}
	R_{ y } ( \theta )
	=
	\begin{bmatrix}
		\cos \frac { \theta } { 2 } & - \sin \frac { \theta } { 2 }
		\\
		\sin \frac { \theta } { 2 } & \phantom{-} \cos \frac { \theta } { 2 }
	\end{bmatrix}
	\ ,
\end{align}

rotates the state of the coin around the $y$ axis of the Bloch sphere through an angle $\frac { \theta } { 2 }$ counterclockwise. Figure \ref{fig: R12piBy16QC} illustrates a quantum circuit that employs the rotation gate $R_{ y } (  \frac { 12 \pi } { 16 } )$ to send the two coins to the state described by Figure \ref{fig: R12piBy16SV}.

\begin{tcolorbox}
	[
		grow to left by = 1.50 cm,
		grow to right by = 1.50 cm,
		colback = white,	
		enhanced jigsaw,	
		sharp corners,
		toprule = 1.0 pt,
		bottomrule = 1.0 pt,
		leftrule = 0.1 pt,
		rightrule = 0.1 pt,
		sharp corners,
		center title,
		fonttitle = \bfseries
	]
	\begin{figure}[H]
		\begin{minipage}[t]{0.45 \textwidth}
			\centering
			\includegraphics[scale = 0.40, trim = {0cm 0cm 0cm 0cm}, clip]{"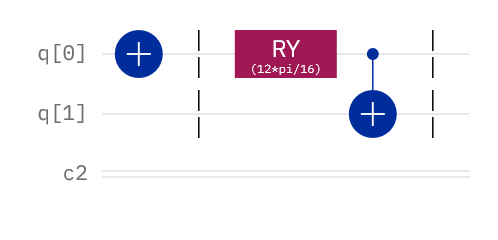"}
			\caption{Given that the initial state is $\ket{ 0 } \ket{ 1 }$, the above quantum circuit applies the rotation gate $R_{ y } (  \frac { 12 \pi } { 16 } )$ to drive the system to the state visualized in Figure \ref{fig: R12piBy16SV}.}
			\label{fig: R12piBy16QC}
		\end{minipage}
		\hfill
		\begin{minipage}[t]{0.5 \textwidth}
			\centering
			\includegraphics[scale = 0.22, trim = {0cm 6cm 10cm 0cm}, clip]{"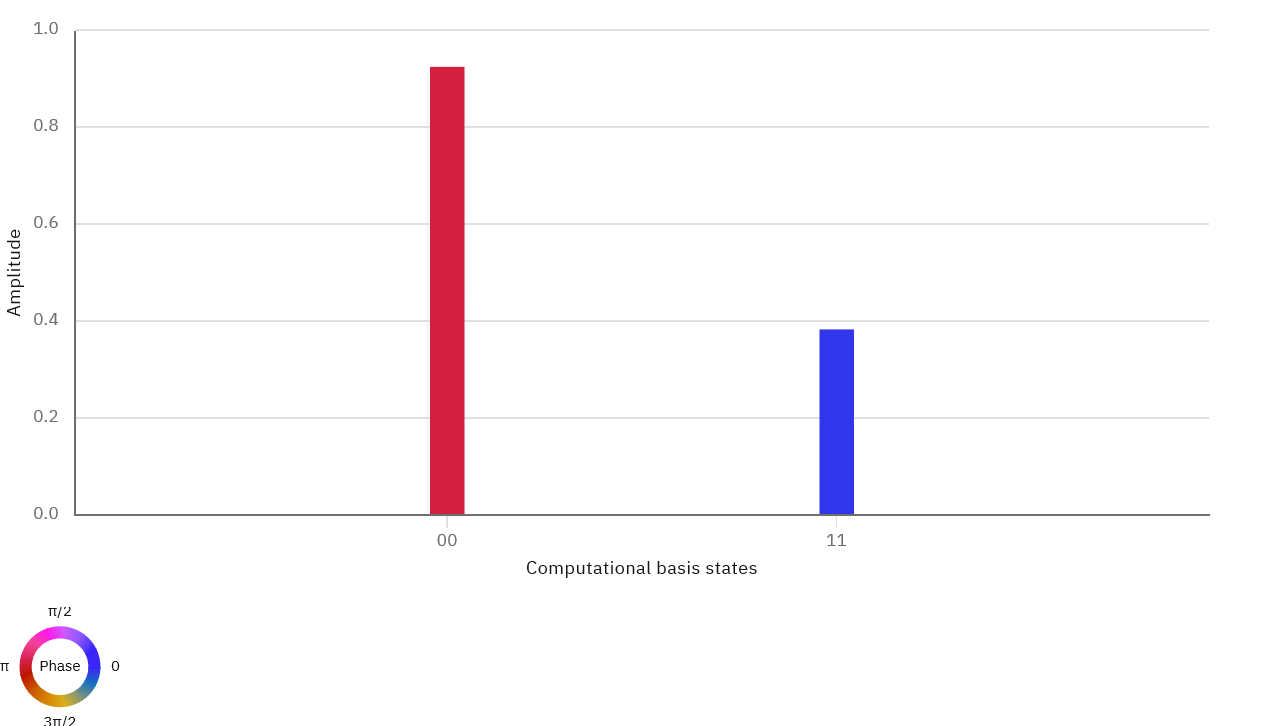"}
			\caption{This circuit of Figure \ref{fig: R12piBy16QC} guarantees that both coins will end up in the same state, but state $\ket{ 0 } \ket{ 0 }$ is much more probable.}
			\label{fig: R12piBy16SV}
		\end{minipage}
	\end{figure}
\end{tcolorbox}

In this scenario, the final state of the two coins is determined by measurements in the computational basis $B = \{ \ket{ 0 }, \ket{ 1 } \}$. The outcome that stands out as particularly probable has player $1$ emerging as the victor and player $2$ experiencing a loss. These measurements and their associated probabilities determine the winner and the loser in the game. Obviously, in such a case, player $2$, being rational, has no incentive whatsoever to participate in the game. There is, however, a way to make the game suspenseful again by letting player $2$ choose the basis with respect to which the coins are measured.

\begin{tcolorbox}
	[
		grow to left by = 1.50 cm,
		grow to right by = 1.50 cm,
		colback = white,	
		enhanced jigsaw,	
		sharp corners,
		toprule = 1.0 pt,
		bottomrule = 1.0 pt,
		leftrule = 0.1 pt,
		rightrule = 0.1 pt,
		sharp corners,
		center title,
		fonttitle = \bfseries
	]
	\begin{figure}[H]
		\begin{minipage}[t]{0.45 \textwidth}
			\centering
			\includegraphics[scale = 0.30, trim = {0cm 0cm 0cm 0cm}, clip]{"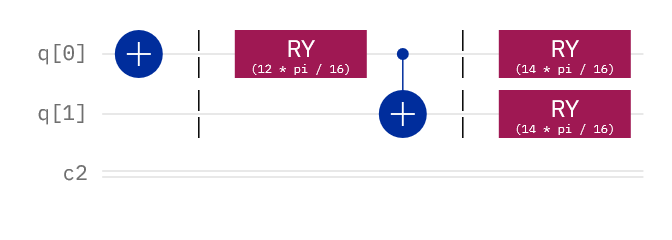"}
			\caption{The effect of the application of the two rotation gates prior to measurement is that the measuring basis is no longer the computational basis.}
			\label{fig: B3piBy4QC}
		\end{minipage}
		\hfill
		\begin{minipage}[t]{0.52 \textwidth}
			\centering
			\includegraphics[scale = 0.22, trim = {0cm 6cm 8cm 0cm}, clip]{"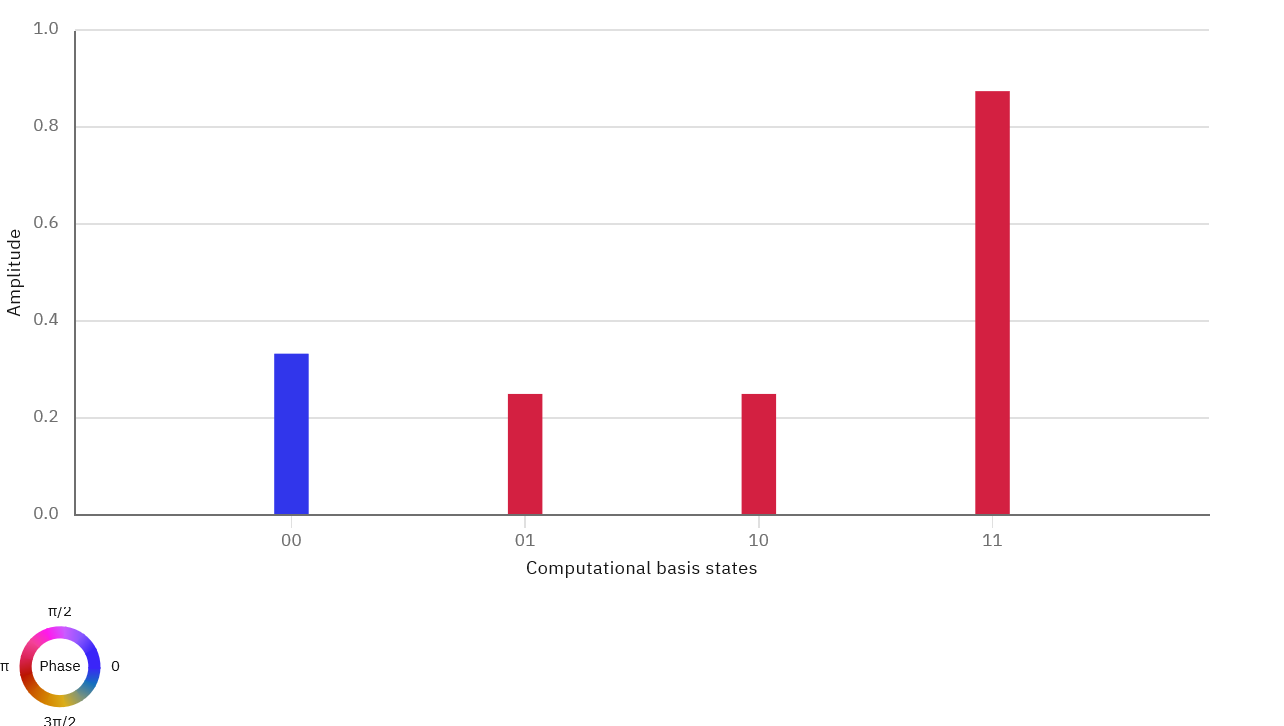"}
			\caption{The circuit of Figure \ref{fig: B3piBy4QC} drives the system to the state described by the above state vector. Now the odds favor player $2$.}
			\label{fig: B3piBy4SV}
		\end{minipage}
	\end{figure}
\end{tcolorbox}

As demonstrated by the above Figures \ref{fig: B3piBy4QC} and \ref{fig: B3piBy4SV}, if player $2$ chooses the proper measurement basis, despite the rotation used by player $1$, he can turn the tables and make sure that the odds of winning are now in his favor.

\section{The formal presentation of the QTG} \label{sec: The Formal Presentation of the QTG}

This section contains the precise description of the QTG, together with its comprehensive mathematical analysis. In our vision of the QTG, rotations play a decisive role, and, since rotations form groups, and group are a prevalent theme in quantum mechanics, we shall employ some elementary concepts from the theory of groups, adhering to the notation and definitions found in standard textbooks such as \cite{Gallian2021, Artin2011} and \cite{Dummit2004}. So, without further ado, we proceed to introduce the required machinery.

\begin{enumerate} [ left = 0.50 cm, labelsep = 1.00 cm ]
	\renewcommand\labelenumi{(\textbf{N}$_\theenumi$)}
	\item	The general cyclic rotation group of order $n$, denoted by $\langle r \rangle$, contains the following $n$ elements
			\begin{align} \label{def: General Cyclic Rotation Group Elements}
				\langle r \rangle
				=
				\{ \mathds{ 1 }, r, r^{ 2 }, \dots, r^{ n - 1 } \}
				\ ,
			\end{align}
	where $r$ is the rotation on the plane by $\frac { 2 \pi } { n }$. Its elements $\mathds{ 1 }, r, r^{ 2 }, \dots, r^{ n - 1 }$ are plane rotations. Specifically, $r^{ k }$ is the rotation by $\frac { 2 \pi k } { n }$, $0 \leq k \leq n - 1$, and $\mathds{ 1 } = r^{ 0 }$ is the identity element.
	\item	To each group element $r^{ k }$, which represents the plane rotation by $\frac { 2 \pi k } { n }$, we associate the rotation gate $R_{ y } ( \frac { 2 \pi k } { n } )$
			\begin{align} \label{eq: Ry Rotation Operator Corresponding to Rotakion r^k}
				R_{ y } \left( \frac { 2 \pi k } { n } \right)
				=
				\begin{bmatrix}
					\cos \frac { \pi k } { n } & - \sin \frac { \pi k } { n }
					\\
					\sin \frac { \pi k } { n } & \phantom{-} \cos \frac { \pi k } { n }
				\end{bmatrix}
				\ ,
			\end{align}
	the two orthonormal kets

			\begin{minipage}[b]{0.40 \textwidth}
				\begin{align} \label{eq: Rotation of Ket 0 By k Pi}
					\ket{ 0 }_{ \frac { 2 \pi k } { n } }
					=
					\begin{bmatrix}
						\cos \frac { \pi k } { n }
						\\
						\sin \frac { \pi k } { n }
					\end{bmatrix}
				\end{align}
			\end{minipage} 
			\hfill
			\begin{minipage}[b]{0.40 \textwidth}
				\begin{align} \label{eq: Rotation of Ket 1 By k Pi}
					\ket{ 1 }_{ \frac { 2 \pi k } { n } }
					=
					\begin{bmatrix}[r]
						- \sin \frac{ \pi k } { n }
						\\
						\phantom{-} \cos \frac { \pi k } { n }
					\end{bmatrix}
					\ ,
				\end{align}
			\end{minipage}

	and the change of basis matrix
			\begin{align} \label{eq: Rotation of Computational Basis By k Pi}
				B_{ \frac { 2 \pi k } { n } }
				=
				\begin{bmatrix}
					\phantom{-} \cos \frac { \pi k } { n } & \sin \frac{ \pi k } { n }
					\\
					- \sin \frac { \pi k } { n } & \cos \frac{ \pi k } { n }
				\end{bmatrix}
				\ ,
			\end{align}
	which affects the change from the computational basis $B = \{ \ket{ 0 }, \ket{ 1 } \}$ to a new basis where the basis kets are $\ket{ 0 }_{ \frac { 2 \pi k } { n } }$ and $\ket{ 1 }_{ \frac { 2 \pi k } { n } }$. More details on the computation of the change of basis matrix can be found in \cite{LeonardS.Woody2022, Williams2011, McMahon2008}.
\end{enumerate}

Before proceeding to formally state the rules of the QTG, let us recall that in all quantum circuits and in all figures of this work, the qubit numbering follows adheres to the IBM Quantum Composer and Qiskit convention, where the rightmost (or top) qubit is the least significant qubit and the leftmost (or bottom) qubit is the most significant qubit.

\begin{enumerate} [ left = 0.50 cm, labelsep = 1.00 cm ]
	\renewcommand\labelenumi{(\textbf{R}$_\theenumi$)}
	\item	The initial state of the two coin system is $\ket{ 0 } \ket{ 1 }$. The winning state for player $1$ (the tosser) is $\ket{ 0 } \ket{ 0 }$, and the winning state for player $2$ (the gambler) is $\ket{ 1 } \ket{ 1 }$.
	\item	Both players agree on a specific cyclic rotation group $\langle r \rangle$ of order $n$, for some sufficiently large $n$. This will be the group upon which the QTG will be based. Both players will pick elements from this group to realize their actions.
	\item	Player $1$ chooses an element $r^{ k } \in \langle r \rangle$, $0 \leq k \leq n - 1$, and uses the rotation operator $R_{ y } ( \frac { 2 \pi k } { n } )$ (recall equation \eqref{eq: Ry Rotation Operator Corresponding to Rotakion r^k}) to entangle the two coins. Formally, player $1$'s action is described by
	\begin{align} \label{eq: General Action of Player 1}
		CNOT \ \cdot \ \left( I \otimes R_{ y } \left( \frac { 2 \pi k } { n } \right) \right)
		\ ,
	\end{align}
	where $CNOT$ and $I$ are the controlled NOT and identity operators, respectively.
	\item	Player $2$ chooses an element $r^{ l } \in \langle r \rangle$, $0 \leq l \leq n - 1$, and measures the two coins with respect to the orthonormal basis $B_{ \frac { 2 \pi l } { n } }$ (recall formula \eqref{eq: Rotation of Computational Basis By k Pi}). Formally, player $2$'s action is described by
	\begin{align} \label{eq: General Action of Player 2}
		R_{ y } \left( \frac { 2 \pi l } { n } \right) \otimes R_{ y } \left( \frac { 2 \pi l } { n } \right)
		\ .
	\end{align}
\end{enumerate}

For concreteness, we have set as $\ket{ 0 } \ket{ 1 }$ the initial state of the game. It is noteworthy to mention that the choice of the initial state does not really matter. Picking a different initial state, would not change the nature of the QTG and the forthcoming analysis.

The QTG is captured by the abstract quantum circuit shown below in Figure \ref{fig: The QC for the QTG}.

\begin{tcolorbox}
	[
		grow to left by = 0.00 cm,
		grow to right by = 0.00 cm,
		colback = MagentaVeryLight!03,			
		enhanced jigsaw,						
		sharp corners,
		toprule = 1.0 pt,
		bottomrule = 1.0 pt,
		leftrule = 0.1 pt,
		rightrule = 0.1 pt,
		sharp corners,
		center title,
		fonttitle = \bfseries
	]
	\begin{figure}[H]
		\centering
		\begin{tikzpicture}[ scale = 0.90 ]
			\begin{yquant}
				nobit AUX_C_0_0;
				nobit AUX_C_0_1;
				qubit { $q_{ 0 }$: \ $\ket{ 1 }$ } Coin_0;
				nobit AUX_C_0_2;
				nobit AUX_C_1_0;
				qubit { $q_{ 1 }$: \ $\ket{ 0 }$ } Coin_1;
				nobit AUX_C_1_1;
				nobit AUX_C_1_2;
				[ name = Ph0, WordBlueVeryLight, line width = 0.50 mm, label = Initial State ]
				barrier ( - ) ;
				hspace {0.250 cm} Coin_0;
				[ draw = RedPurple!50, fill = RedPurple!50, radius = 0.5 cm ] box {\sf{RY}$( \frac { 2 \pi k } { n } )$} Coin_0;
				hspace {0.250 cm} Coin_0;
				cnot Coin_1 | Coin_0;
				hspace {0.250 cm} Coin_0;
				[ name = Ph1, WordBlueVeryLight, line width = 0.50 mm, label = Phase $1$ ]
				barrier ( - ) ;
				hspace {0.50 cm} Coin_0;
				[ draw = GreenLighter2!50, fill = GreenLighter2!50, radius = 0.5 cm ] box {\sf{RY}$( \frac { 2 \pi l } { n } )$} Coin_0;
				hspace {0.250 cm} Coin_0;
				hspace {0.50 cm} Coin_1;
				[ draw = GreenLighter2!50, fill = GreenLighter2!50, radius = 0.5 cm ] box {\sf{RY}$( \frac { 2 \pi l } { n } )$} Coin_1;
				hspace {0.250 cm} Coin_1;
				[ name = Ph2, WordBlueVeryLight, line width = 0.50 mm, label = Phase $2$ ]
				barrier ( - ) ;
				[ draw = white, fill = MagentaDark ] [ radius = 0.5 cm ] measure Coin_0;
				[ draw = white, fill = MagentaDark ] [ radius = 0.5 cm ] measure Coin_1;
				hspace {0.5 cm} Coin_0;
				output {$\ket{ ? }$} Coin_0;
				output {$\ket{ ? }$} Coin_1;
				\node
				[
				alice,
				scale = 1.25,
				anchor = center,
				above right = 1.50 cm and 1.750 cm of Ph0,
				label = { [ label distance = 0.00 cm ] north: { \color{RedPurple}\textbf{Tosser} } }
				]
				( ) { \footnotesize { \color{RedPurple}Entangles the coins} };
				\node
				[
				charlie,
				scale = 1.25,
				anchor = center,
				above right = 1.50 cm and 1.50 cm of Ph1,
				label = { [ label distance = 0.00 cm ] north: { \color{GreenLighter2}\textbf{Gambler} } }
				]
				( ) { \footnotesize { \color{GreenLighter2}Changes the basis} };
				\node [ below = 2.50 cm ] at (Ph0) { $\ket{ \psi_{ 0 } }$ };
				\node [ below = 2.50 cm ] at (Ph1) { $\ket{ \psi_{ 1 } }$ };
				\node [ below = 2.50 cm ] at (Ph2) { $\ket{ \psi_{ 2 } }$ };
			\end{yquant}
		\end{tikzpicture}
		\caption{This figure shows the quantum circuit that corresponds to the QRG.}
		\label{fig: The QC for the QTG}
	\end{figure}
\end{tcolorbox}

\section{The mathematical analysis of the QTG} \label{sec: The Mathematical Analysis of the QTG}

In this section we shall analyze the QTG from a rigorous mathematical viewpoint, in order to draw conclusions about the options of the two players and their odds of winning. We use the abstract quantum circuit of Figure \ref{fig: The QC for the QTG} as a guide. In our exposition, we follow the Qiskit \cite{Qiskit2023} way of ordering the qubits, which means that the bottom qubit is the most significant and the top the least significant.

By rule (\textbf{R}$_{ 1 }$), the initial state of the system, denoted by $\ket{ \psi_{ 0 } }$, is given by

\begin{align} \label{eq: TQG Initial State}
	\ket{ \psi_0 }
	=
	\ket{ 0 }
	\ket{ 1 }
	=
	\begin{bmatrix}
		0
		\\
		1
		\\
		0
		\\
		0
	\end{bmatrix}
	\ .
\end{align}

Player $1$, the tosser, acts first on the system, using the action descried by equation \eqref{eq: General Action of Player 1} in rule (\textbf{R}$_{ 3 }$). In view of formula \eqref{eq: Ry Rotation Operator Corresponding to Rotakion r^k}, the matrix representation of the tensor product $( I \otimes R_{ y } ( \frac { 2 \pi k } { n } ) )$ is

{\small
	\begin{align} \label{eq: I Tensor Product Ry}
		I \otimes R_{ y } \left( \frac { 2 \pi k } { n } \right)
		\overset { \eqref{eq: Ry Rotation Operator Corresponding to Rotakion r^k} } { = }
		\begin{bmatrix}
			1 & 0
			\\
			0 & 1
		\end{bmatrix}
		\otimes
		\begin{bmatrix}
			\cos \frac { \pi k } { n } & - \sin \frac { \pi k } { n } \\
			\sin \frac { \pi k } { n } & \phantom{-} \cos \frac { \pi k } { n }
		\end{bmatrix}
		=
		\begin{bmatrix}
			\cos \frac { \pi k } { n } & - \sin \frac { \pi k } { n } & 0 & 0
			\\
			\sin \frac { \pi k } { n } & \phantom{-} \cos \frac { \pi k } { n } & 0 & 0
			\\
			0 & 0 & \cos \frac { \pi k } { n } & - \sin \frac { \pi k } { n }
			\\
			0 & 0 & \sin \frac { \pi k } { n } & \phantom{-} \cos \frac { \pi k } { n }
		\end{bmatrix}
		\ .
	\end{align}
}

In this work we follow the Qiskit conventions in the numbering of qubits. Therefore, the matrix representation of the $CNOT$, as applied in the circuit of Figure \ref{fig: The QC for the QTG}, is the one given below, according to the Qiskit documentation \cite{QiskitCNOTMatrix2023}.

\begin{align} \label{eq: CNOT Operator}
	CNOT
	=
	\begin{bmatrix}
		1 & 0 & 0 & 0
		\\
		0 & 0 & 0 & 1
		\\
		0 & 0 & 1 & 0
		\\
		0 & 1 & 0 & 0
	\end{bmatrix}
	\
\end{align}

Using \eqref{eq: I Tensor Product Ry} and \eqref{eq: CNOT Operator}, we may derive the matrix representation of the action of player $1$.

{\small
	\begin{align} \label{eq: Matrix Form of the Action of Player 1}
		CNOT \ \cdot \ \left( I \otimes R_{ y } \left( \frac { 2 \pi k } { n } \right) \right)
		&\overset { \eqref{eq: I Tensor Product Ry}, \eqref{eq: CNOT Operator} } { = }
		\begin{bmatrix}
			1 & 0 & 0 & 0
			\\
			0 & 0 & 0 & 1
			\\
			0 & 0 & 1 & 0
			\\
			0 & 1 & 0 & 0
		\end{bmatrix}
		\
		\begin{bmatrix}
			\cos \frac { \pi k } { n } & - \sin \frac { \pi k } { n } & 0 & 0
			\\
			\sin \frac { \pi k } { n } & \phantom{-} \cos \frac { \pi k } { n } & 0 & 0
			\\
			0 & 0 & \cos \frac { \pi k } { n } & - \sin \frac { \pi k } { n }
			\\
			0 & 0 & \sin \frac { \pi k } { n } & \phantom{-} \cos \frac { \pi k } { n }
		\end{bmatrix}
		\nonumber \\
		&=
		\begin{bmatrix}]
			\cos \frac { \pi k } { n } & - \sin \frac { \pi k } { n } & 0 & 0
			\\
			0 & 0 & \sin \frac { \pi k } { n } & \phantom{-} \cos \frac { \pi k } { n }
			\\
			0 & 0 & \cos \frac { \pi k } { n } & - \sin \frac { \pi k } { n }
			\\
			\sin \frac { \pi k } { n } & \phantom{-} \cos \frac { \pi k } { n } & 0 & 0
		\end{bmatrix}
		\
	\end{align}
}

By combining \eqref{eq: TQG Initial State} and \eqref{eq: Matrix Form of the Action of Player 1}, we conclude that the state of the two coins at the end of Phase $1$, i.e., after the tosser's action, which is denoted by $\ket{ \psi_{ 1 } }$, is given by

\begin{align} \label{eq: TQG Phase 1}
	\ket{ \psi_{ 1 } }
	=
	\begin{bmatrix}
		- \sin \frac { \pi k }{ n }
		\\
		0
		\\
		0
		\\
		\phantom{-} \cos \frac { \pi k } { n }
	\end{bmatrix}
	\ .
\end{align}

\begin{example} [Playing in the group $\langle \frac { 2 \pi } { 16 } \rangle$] \label{xmp: Playing in the Group 2 pi By 16 Phase 1}
	This example examines how the QTG evolves up to Phase $1$ (see Figure \ref{fig: The QC for the QTG}), when the game is played in the cyclic rotation group $\langle \frac { 2 \pi } { 16 } \rangle$. Practically, this implies that both players may choose for their action one the $16$ rotations of this group.

	The general form of the action of Player $1$ is given by \eqref{eq: Matrix Form of the Action of Player 1} and the state of the two coins by \eqref{eq: TQG Phase 1}. If the game were to end at this point in time, measuring the two coins would find the system in state $\ket{ 0 } \ket{ 0 }$ with probability $\sin^{ 2 } \frac { \pi k }{ n }$, and in state $\ket{ 1 } \ket{ 1 }$ with probability $\cos^{ 2 } \frac { \pi k }{ n }$. Hence, the probability that the tosser would win ranges from $0$, in case the tosser uses the identity element of the group, namely the $0$ radians rotation, to $1$, in case he uses the $\frac { \pi }{ 2 }$ radians rotation. For every possible choice of rotation from the $\langle \frac { 2 \pi } { 16 } \rangle$ group, the corresponding probability for player $1$ is given by the red line of Figure \ref{fig: Phase 1 Probabilities}. Symmetrically, for every possible choice of rotation from $\langle \frac { 2 \pi } { 16 } \rangle$, the corresponding probability for player $2$ is given by the green line of Figure \ref{fig: Phase 1 Probabilities}. The specific numerical values of the players' probabilities are contained in Table \ref{tbl: Phase 1 Probabilities}.
	\hfill $\triangleleft$
\end{example}

\begin{tcolorbox}
	[
		grow to left by = 1.00 cm,
		grow to right by = 1.00 cm,
		colback = MagentaVeryLight!03,			
		enhanced jigsaw,						
		sharp corners,
		toprule = 1.0 pt,
		bottomrule = 1.0 pt,
		leftrule = 0.1 pt,
		rightrule = 0.1 pt,
		sharp corners,
		center title,
		fonttitle = \bfseries
	]
	\begin{figure}[H]
		\centering
		\pgfplotsset { compat = 1.16 }
		\begin{tikzpicture} [ scale = 1.00 ]
			\begin{axis}
				[
					width = 7.00 cm, height = 5.50 cm,
					axis line style = { WordBlueDarker, line width = 0.25 mm },
					axis x line = center,
					axis y line = middle,
					axis line style = { ->, > = { latex [ length = 7pt ] } },
					xmin = - 0.1, xmax = 3.40,
					ymin = - 0.1, ymax = 1.1,
					xtick =
					{
						pi * 1 / 16, pi * 2 / 16, pi * 3 / 16, pi * 4 / 16,
						pi * 5 / 16, pi * 6 / 16, pi * 7 / 16, pi * 8 / 16,
						pi * 9 / 16, pi * 10 / 16, pi * 11 / 16, pi * 12 / 16,
						pi * 13 / 16, pi * 14 / 16, pi * 15 / 16, pi * 16 / 16
					},
					xticklabels =
					{
						$\frac { \pi } { 16 }$, $\frac { 2 \pi } { 16 }$,
						$\frac { 3 \pi } { 16 }$, $\frac { \pi } { 4 }$,
						$\frac { 5 \pi } { 16 }$, $\frac { 6 \pi } { 16 }$,
						$\frac { 7 \pi } { 16 }$, $\frac { \pi } { 2 }$,
						$\frac { 9 \pi } { 16 }$, $\frac { 10 \pi } { 16 }$,
						$\frac { 11 \pi } { 16 }$, $\frac { 3 \pi } { 4 }$,
						$\frac { 13 \pi } { 16 }$, $\frac { 14 \pi } { 16 }$,
						$\frac { 15 \pi } { 16 }$, $\pi$
					},
					ytick = { 0, 0.125, 0.25, 0.375, 0.50, 0.625, 0.75, 0.875, 1.00 },
					yticklabels = { 0, 0.125, 0.250, 0.375, 0.500, 0.625, 0.750, 0.875, 1 },
					x label style =
					{
						at = { ( axis description cs: 0.5, -0.075 ) }, fill = WordVeryLightTeal, anchor = north
					},
					xlabel = { \textbf{Rotation angle} $\frac { \pi k } { 16 } \ (0 \leq k \leq 15)$ },
					y label style =
					{
						at = { ( axis description cs: 0.025, 1.040 ) }, fill = WordVeryLightTeal, anchor = south
					},
					ylabel = { \textbf{Probabilities} },
					x post scale = 2.0,
					y post scale = 1.5,
					grid = minor,
					trig format plots = rad,
					domain = 0 * pi : pi,
					samples = 100,
					declare function = { t(\x) = sin( \x )^2; g(\x) = cos( \x )^2; },
					legend style = { at = { ( 0.5, 1.2 ) }, anchor = south }
				]
				\addlegendimage { empty legend }
				\addlegendentry { \textbf{Phase} $1$ in $\langle \frac { 2 \pi } { 16 } \rangle$ }
				\addplot [ RedPurple, line width = 0.4 mm ] { t(\x) };
				\node [ pin = 30 : { \color{RedPurple} \textbf{Tosser} } ] at ( 10 * pi / 16, 0.855 ) {};
				\addlegendentry { Probability Player $1$ wins }
				\addplot [ GreenLighter2, line width = 0.4 mm ] { g ( \x ) };
				\node [ pin = 30 : { \color{GreenLighter2} \textbf{Gambler} } ] at ( pi / 8, 0.855 ) {};
				\addlegendentry { Probability Player $2$ wins }
				\addplot [ color = RedPurple, draw = none, mark size = 3 pt, mark = * ] coordinates
				{
					( 0, { t ( 0 ) } )
					( pi * 1 / 16, { t ( pi * 1 / 16 ) } )
					( pi * 2 / 16, { t ( pi * 2 / 16 ) } )
					( pi * 3 / 16, { t ( pi * 3 / 16 ) } )
					( pi * 4 / 16, { t ( pi * 4 / 16 ) } )
					( pi * 5 / 16, { t ( pi * 5 / 16 ) } )
					( pi * 6 / 16, { t ( pi * 6 / 16 ) } )
					( pi * 7 / 16, { t ( pi * 7 / 16 ) } )
					( pi * 8 / 16, { t ( pi * 8 / 16 ) } )
					( pi * 9 / 16, { t ( pi * 9 / 16 ) } )
					( pi * 10 / 16, { t ( pi * 10 / 16 ) } )
					( pi * 11 / 16, { t ( pi * 11 / 16 ) } )
					( pi * 12 / 16, { t ( pi * 12 / 16 ) } )
					( pi * 13 / 16, { t ( pi * 13 / 16 ) } )
					( pi * 14 / 16, { t ( pi * 14 / 16 ) } )
					( pi * 15 / 16, { t ( pi * 15 / 16 ) } )
					( pi * 16 / 16, { t ( pi * 16 / 16 ) } )
				};
				\addplot [ color = GreenLighter2, draw = none, mark size = 3 pt, mark = square* ] coordinates
				{
					( 0, { g ( 0 ) } )
					( pi * 1 / 16, { g ( pi * 1 / 16 ) } )
					( pi * 2 / 16, { g ( pi * 2 / 16 ) } )
					( pi * 3 / 16, { g ( pi * 3 / 16 ) } )
					( pi * 4 / 16, { g ( pi * 4 / 16 ) } )
					( pi * 5 / 16, { g ( pi * 5 / 16 ) } )
					( pi * 6 / 16, { g ( pi * 6 / 16 ) } )
					( pi * 7 / 16, { g ( pi * 7 / 16 ) } )
					( pi * 8 / 16, { g ( pi * 8 / 16 ) } )
					( pi * 9 / 16, { g ( pi * 9 / 16 ) } )
					( pi * 10 / 16, { g ( pi * 10 / 16 ) } )
					( pi * 11 / 16, { g ( pi * 11 / 16 ) } )
					( pi * 12 / 16, { g ( pi * 12 / 16 ) } )
					( pi * 13 / 16, { g ( pi * 13 / 16 ) } )
					( pi * 14 / 16, { g ( pi * 14 / 16 ) } )
					( pi * 15 / 16, { g ( pi * 15 / 16 ) } )
					( pi * 16 / 16, { g ( pi * 16 / 16 ) } )
				};
			\end{axis}
		\end{tikzpicture}
		\caption{This figure shows both players probabilities to win at the end of Phase $1$ when the QTG is played in $\langle \frac { 2 \pi } { 16 } \rangle$. The red and green lines depict the tosser and the gambler's, respectively, probabilities to win.}
		\label{fig: Phase 1 Probabilities}
	\end{figure}
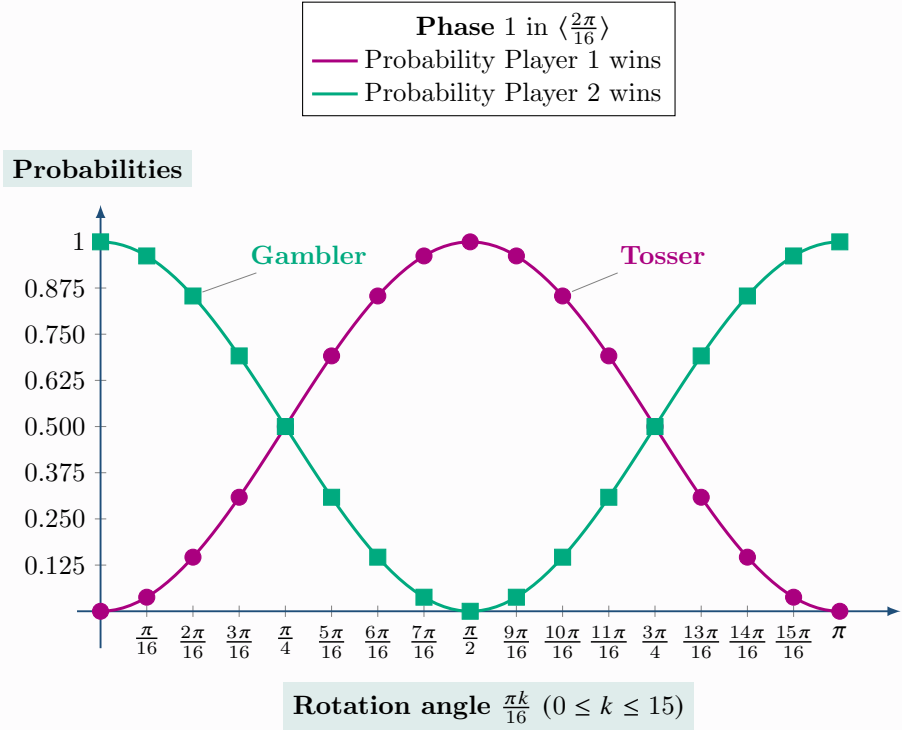
\end{tcolorbox}

\begin{tcolorbox}
	[
		grow to left by = 1.00 cm,
		grow to right by = 1.00 cm,
		colback = WordIceBlue!25,	%
		enhanced jigsaw,			
		sharp corners,
		boxrule = 0.1 pt,
		toprule = 0.1 pt,
		bottomrule = 0.1 pt
	]
	\begin{table}[H]
		\renewcommand{\arraystretch}{1.60}
		\caption{This table gives the probabilities for each player at the end of Phase $1$ when QTG is played in $\langle \frac { 2 \pi } { 16 } \rangle$.}
		\label{tbl: Phase 1 Probabilities}
		\centering
		{\small
		\begin{tabular}
			{
				>{\centering\arraybackslash} m{1.50 cm} !{\vrule width 0.25 pt}
				>{\centering\arraybackslash} m{0.85 cm} !{\vrule width 0.25 pt}
				>{\centering\arraybackslash} m{0.85 cm} !{\vrule width 0.25 pt}
				>{\centering\arraybackslash} m{0.85 cm} !{\vrule width 0.25 pt}
				>{\centering\arraybackslash} m{0.85 cm} !{\vrule width 0.25 pt}
				>{\centering\arraybackslash} m{0.85 cm} !{\vrule width 0.25 pt}
				>{\centering\arraybackslash} m{0.85 cm} !{\vrule width 0.25 pt}
				>{\centering\arraybackslash} m{0.85 cm} !{\vrule width 0.25 pt}
				>{\centering\arraybackslash} m{0.85 cm} !{\vrule width 0.25 pt}
				>{\centering\arraybackslash} m{0.85 cm}
			}
			\Xhline{4\arrayrulewidth}
			\multicolumn{10}{c}
			{ \cellcolor[HTML]{000000} { \color[HTML]{FFFFFF} Probabilities at the end of Phase $1$ in $\langle \frac { 2 \pi } { 16 } \rangle$ } }
			\\
			\Xhline{\arrayrulewidth}
			{ \cellcolor[HTML]{000000} { \color[HTML]{FFFFFF} $k$ } }
			&
			{ \cellcolor[HTML]{000000} { \color[HTML]{FFFFFF} $0$ } }
			&
			{ \cellcolor[HTML]{000000} { \color[HTML]{FFFFFF} $1, 15$ } }
			&
			{ \cellcolor[HTML]{000000} { \color[HTML]{FFFFFF} $2, 14$ } }
			&
			{ \cellcolor[HTML]{000000} { \color[HTML]{FFFFFF} $3, 13$ } }
			&
			{ \cellcolor[HTML]{000000} { \color[HTML]{FFFFFF} $4, 12$ } }
			&
			{ \cellcolor[HTML]{000000} { \color[HTML]{FFFFFF} $5, 11$ } }
			&
			{ \cellcolor[HTML]{000000} { \color[HTML]{FFFFFF} $6, 10$ } }
			&
			{ \cellcolor[HTML]{000000} { \color[HTML]{FFFFFF} $7, 9$ } }
			&
			{ \cellcolor[HTML]{000000} { \color[HTML]{FFFFFF} $8$ } }
			\\
			\Xhline{3\arrayrulewidth}
			{ \cellcolor[HTML]{000000} { \color[HTML]{FFFFFF} Tosser } }
			&
			$0$
			&
			$0.038$
			&
			$0.146$
			&
			$0.309$
			&
			$0.5$
			&
			$0.691$
			&
			$0.854$
			&
			$0.962$
			&
			$1$
			\\
			\Xhline{\arrayrulewidth}
			{ \cellcolor[HTML]{000000} { \color[HTML]{FFFFFF} Gambler } }
			&
			$1$
			&
			$0.962$
			&
			$0.854$
			&
			$0.691$
			&
			$0.5$
			&
			$0.309$
			&
			$0.146$
			&
			$0.038$
			&
			$0$
			\\
			\Xhline{4\arrayrulewidth}
		\end{tabular}
		}
		\renewcommand{\arraystretch}{1.0}
	\end{table}
\end{tcolorbox}

Continuing with our analysis, let us now proceed to Phase $2$ (see Figure \ref{fig: The QC for the QTG}). During Phase $2$, player $2$, the gambler, acts by changing the measurement basis. This action is expressed by the tensor product $R_{ y } ( \frac { 2 \pi l } { n } ) \otimes R_{ y } ( \frac { 2 \pi l } { n } )$ of equation \eqref{eq: Rotation of Computational Basis By k Pi}, which is given below in explicit form.

{\small
	\begin{align} \label{eq: Ry Tensor Product Ry}
		R_{ y } \left( \frac { 2 \pi l } { n } \right) \otimes R_{ y } \left( \frac { 2 \pi l } { n } \right)
		&\overset { \eqref{eq: General Action of Player 2} } { = }
		\begin{bmatrix}
			\cos \frac { \pi l } { n } & - \sin \frac { \pi l } { n } \\
			\sin \frac { \pi l } { n } & \cos \frac { \pi l } { n }
		\end{bmatrix}
		\otimes
		\begin{bmatrix}
			\cos \frac { \pi l } { n } & - \sin \frac { \pi l } { n } \\
			\sin \frac { \pi l } { n } & \cos \frac { \pi l } { n }
		\end{bmatrix}
		\nonumber \\
		&=
		\begin{bmatrix}
			\cos \frac { \pi l } { n } \cos \frac { \pi l } { n } &
			- \cos \frac { \pi l } { n } \sin \frac { \pi l } { n } &
			- \sin \frac{ \pi l } { n } \cos \frac { \pi l } { n } &
			\phantom{-} \sin \frac { \pi l } { n } \sin \frac { \pi l } { n }
			\\
			\cos \frac { \pi l } { n } \sin \frac { \pi l } { n } &
			\phantom{-} \cos \frac { \pi l } { n } \cos \frac { \pi l } { n } &
			- \sin \frac{ \pi l }{ n } \sin \frac { \pi l } { n } &
			- \sin \frac { \pi l } { n } \cos \frac { \pi l } { n }
			\\
			\sin \frac { \pi l } { n } \cos \frac { \pi l } { n } &
			- \sin \frac { \pi l } { n } \sin \frac{ \pi l } { n } &
			\phantom{-} \cos \frac { \pi l } { n } \cos \frac { \pi l } { n } &
			- \cos \frac { \pi l } { n } \sin \frac { \pi l } { n }
			\\
			\sin \frac { \pi l } { n } \sin \frac { \pi l } { n } &
			\phantom{-} \sin \frac { \pi l } { n } \cos \frac { \pi l } { n } &
			\phantom{-} \cos \frac { \pi l } { n } \sin \frac { \pi l } { n } &
			\phantom{-} \cos \frac { \pi l } { n } \cos \frac { \pi l } { n }
		\end{bmatrix}
		\ .
	\end{align}
}

Therefore, as a result of the gambler's action, the system of the two coins at the end of Phase $2$ will be in the following $\ket{ \psi_{ 2 } }$ state.

{\small
	\begin{align} \label{eq: TQG Phase 2}
		\ket{ \psi_{ 2 } }
		&\overset { \eqref{eq: TQG Phase 1}, \eqref{eq: Ry Tensor Product Ry} } { = }
		\begin{bmatrix}
			\cos \frac { \pi l } { n } \cos \frac { \pi l } { n } &
			- \cos \frac { \pi l } { n } \sin \frac { \pi l } { n } &
			- \sin \frac{ \pi l } { n } \cos \frac { \pi l } { n } &
			\phantom{-} \sin \frac { \pi l } { n } \sin \frac { \pi l } { n }
			\\
			\cos \frac { \pi l } { n } \sin \frac { \pi l } { n } &
			\phantom{-} \cos \frac { \pi l } { n } \cos \frac { \pi l } { n } &
			- \sin \frac{ \pi l }{ n } \sin \frac { \pi l } { n } &
			- \sin \frac { \pi l } { n } \cos \frac { \pi l } { n }
			\\
			\sin \frac { \pi l } { n } \cos \frac { \pi l } { n } &
			- \sin \frac { \pi l } { n } \sin \frac{ \pi l } { n } &
			\phantom{-} \cos \frac { \pi l } { n } \cos \frac { \pi l } { n } &
			- \cos \frac { \pi l } { n } \sin \frac { \pi l } { n }
			\\
			\sin \frac { \pi l } { n } \sin \frac { \pi l } { n } &
			\phantom{-} \sin \frac { \pi l } { n } \cos \frac { \pi l } { n } &
			\phantom{-} \cos \frac { \pi l } { n } \sin \frac { \pi l } { n } &
			\phantom{-} \cos \frac { \pi l } { n } \cos \frac { \pi l } { n }
		\end{bmatrix}
		\
		\begin{bmatrix}
			- \sin \frac { \pi k } { n }
			\\
			0
			\\
			0
			\\
			\cos \frac { \pi k } { n }
		\end{bmatrix}
		\nonumber \\
		&=
		\begin{bmatrix}
			- \sin \frac { \pi k } { n } \cos^{ 2 } \frac { \pi l } { n } +
			\cos \frac { \pi k } { n } \sin^{ 2 } \frac { \pi l } { n }
			\\
			- \sin \frac { \pi k } { n } \cos \frac { \pi l } { n } \sin \frac { \pi l } { n } -
			\cos \frac { \pi k } { n } \cos \frac { \pi l } { n } \sin \frac { \pi l } { n }
			\\
			- \sin \frac { \pi k } { n } \cos \frac { \pi l } { n } \sin \frac { \pi l } { n } -
			\cos \frac { \pi k } { n } \cos \frac { \pi l } { n } \sin \frac { \pi l } { n }
			\\
			- \sin \frac { \pi k } { n } \sin^{ 2 } \frac { \pi l } { n } +
			\cos \frac { \pi k } { n } \cos^{ 2 } \frac { \pi l } { n }
		\end{bmatrix}
		\ .
	\end{align}
}

After the completion of Phase $2$, the measurement takes place. As formula \eqref{eq: TQG Phase 2} implies, the probability $P_{ T }$ of player $1$ winning, that is of observing the system in state $\ket{ 0 } \ket{ 0 }$, is given by the next equation

\begin{align} \label{eq: Probability Tosser Wins}
	P_{ T } ( k, l )
	=
	\left(
	\cos \frac { \pi k } { n } \sin^{ 2 } \frac { \pi l } { n } -
	\sin \frac { \pi k } { n } \cos^{ 2 } \frac { \pi l } { n }
	\right)^{ 2 }
	\ .
\end{align}

In a symmetric way, the probability $P_{ G }$ of player $2$ winning, i.e, finding the system in state $\ket{ 1 } \ket{ 1 }$ is

\begin{align} \label{eq: Probability Gambler Wins}
	P_{ G } ( k, l )
	=
	\left(
	\cos \frac { \pi k } { n } \cos^{ 2 } \frac { \pi l } { n } -
	\sin \frac { \pi k } { n } \sin^{ 2 } \frac { \pi l } { n }
	\right)^{ 2 }
	\ .
\end{align}

Another important remark is that the probability $P_{ \text{draw} }$ that neither player wins, namely that the system ends up in the state $\ket{ 0 } \ket{ 1 }$ or $\ket{ 1 } \ket{ 0 }$, is expressed by

\begin{align} \label{eq: Draw Probability}
	P_{ \text{draw} } ( k, l )
	=
	2 \
	\left(
	\cos \frac { \pi k } { n } + \sin \frac { \pi k } { n }
	\right)^{ 2 }
	\cos^{ 2 } \frac { \pi l } { n } \sin^{ 2 } \frac { \pi l } { n }
	\ ,
\end{align}

and is, in general, nonzero.

%

A close examination of Equations \eqref{eq: Probability Tosser Wins} and \eqref{eq: Probability Gambler Wins} reveals the characteristic property of the QTG.
Irrespective of the tosser's choice of rotation, the subsequent action of the gambler balances the situation and puts both players on a equal footing. For any given $k$, $0 \leq k \leq n - 1$, chosen by the tosser, the gambler picks some orthonomal basis $B_{ \frac { 2 \pi l } { n } }$, $0 \leq l \leq n - 1$, which leads to the probabilities $P_{ T } ( k, l )$ and $P_{ G } ( k, l )$ that express the chances of the former and the latter, respectively, to win the game The crucial observation is that for every basis $B_{ \frac { 2 \pi l } { n } }$, there exists a \emph{dual} basis $B_{ \frac { 2 \pi l^{ \ast } } { n } }$ that exchanges the aforementioned probabilities, so that $P_{ T } ( k, l^{ \ast } ) = P_{ G } ( k, l )$ and $P_{ G } ( k, l^{ \ast } ) = P_{ T } ( k, l )$. For instance, say that for some basis $B_{ \frac { 2 \pi l } { n } }$, $P_{ T } ( k, l ) > P_{ G } ( k, l )$, i.e., player $1$ is more likely to win than player $2$. Then for the dual basis $B_{ \frac { 2 \pi l^{ \ast } } { n } }$, $P_{ G } ( k, l^{ \ast } ) = P_{ T } ( k, l ) > P_{ G } ( k, l ) = P_{ T } ( k, l^{ \ast } )$, which shows that the odds are perfectly reserved and player $2$ is now more likely to win than player $1$ by exactly the same amount. This fact is formally proved in the next Theorem \ref{thr: The QTG is Fair}.

\begin{theorem} \label{thr: The QTG is Fair}
	The characteristic property of the QTG is that for each $k$, $0 \leq k \leq n - 1$, to every orthonomal basis $B_{ \frac { 2 \pi l } { n } }$, $0 \leq l \leq n - 1$, corresponds a \emph{dual} basis $B_{ \frac { 2 \pi l^{ \ast } } { n } }$ such that:
	\begin{align} \label{eq: The Characteristic Property of the QTG}
		P_{ T } ( k, l^{ \ast } ) &= P_{ G } ( k, l )
		\\
		&\text{and}
		\nonumber \\
		P_{ G } ( k, l^{ \ast } ) &= P_{ T } ( k, l )
		\ .
	\end{align}
\end{theorem}

\begin{proof}
	The proof is quite straightforward. First, let us recall some well-known trigonometric identities (see any standard textbook, such as \cite{Beecher2016}).
	\begin{align}
		\cos ( \theta \pm \frac { \pi } { 2 } ) &= \mp \sin \theta
		&
		\sin ( \theta \pm \frac { \pi } { 2 } ) &= \pm \cos \theta
		\tag{$\star$}
	\end{align}
	Let us fix a specific $k$, $0 \leq k \leq n - 1$, representing the plane rotation $r^{ k }$ by $\frac { 2 \pi k } { n }$ radians, which is employed by the tosser during Phase $1$. Let us further assume that during Phase $2$ the gambler picks the orthonomal basis $B_{ \frac { 2 \pi l } { n } }$ for some $l$, $0 \leq l \leq n - 1$. According to formulas \eqref{eq: Probability Tosser Wins} and \eqref{eq: Probability Gambler Wins}, the corresponding probabilities that player $1$ and $2$ wins are
	\begin{align}
		P_{ T } ( k, l )
		&=
		\left(
		\cos \frac { \pi k } { n } \sin^{ 2 } \frac { \pi l } { n } -
		\sin \frac { \pi k } { n } \cos^{ 2 } \frac { \pi l } { n }
		\right)^{ 2 }
		\tag{\ref{eq: Probability Tosser Wins}}
		\\
		&\text{and}
		\nonumber \\
		P_{ G } ( k, l )
		&=
		\left(
		\cos \frac { \pi k } { n } \cos^{ 2 } \frac { \pi l } { n } -
		\sin \frac { \pi k } { n } \sin^{ 2 } \frac { \pi l } { n }
		\right)^{ 2 }
		\ .
		\tag{\ref{eq: Probability Gambler Wins}}
	\end{align}
	Given $l$, we define its \emph{dual}, denoted by $l^{ \ast }$, as follows:
	\begin{align} \label{eq:Dual l^* D}
		l^{ \ast }
		\equiv
		\left( l + \frac { n } { 2 } \right) \bmod n
		\ .
	\end{align}
	The above compact definition, can be cast in the more expedient form shown below.
	\begin{align} \label{eq: Explicit Dual l^*}
		l^{ \ast }
		=
		\left\{
		\begin{matrix}
			l + \frac { n } { 2 } & \text{ if } l + \frac { n } { 2 } < n \\
			l - \frac { n } { 2 } & \text{ if } l + \frac { n } { 2 } \geq n
		\end{matrix}
		\right.
		\ , \ 0 \leq l \leq n - 1 \ .
	\end{align}
	We now proceed with the proof by distinguishing the following cases.
	\begin{itemize}
		\item	$l + \frac { n } { 2 } < n$. In this case, according to \eqref{eq: Explicit Dual l^*}, $l^{ \ast } = l + \frac { n } { 2 }$. Therefore,
				\begin{align}
					P_{ T } ( k, l^{ \ast } )
					&\overset{\eqref{eq: Probability Tosser Wins}}{=}
					\left(
					\cos \frac { \pi k } { n } \sin^{ 2 } \frac { \pi l^{ \ast } } { n } -
					\sin \frac { \pi k } { n } \cos^{ 2 } \frac { \pi l^{ \ast } } { n }
					\right)^{ 2 }
					\nonumber
					\\
					&=
					\left(
					\cos \frac { \pi k } { n } \sin^{ 2 } \left( \frac { 2 \pi l + \pi n} { 2 n } \right) -
					\sin \frac { \pi k } { n } \cos^{ 2 } \left( \frac { 2 \pi l + \pi n} { 2 n } \right)
					\right)^{ 2 }
					\nonumber
					\\
					&=
					\left(
					\cos \frac { \pi k } { n } \sin^{ 2 } \left( \frac { \pi l } { n } + \frac { \pi } { 2 } \right) -
					\sin \frac { \pi k } { n } \cos^{ 2 } \left( \frac { \pi l } { n } + \frac { \pi } { 2 } \right)
					\right)^{ 2 }
					\nonumber
					\\
					&=
					\left(
					\cos \frac { \pi k } { n } \cos^{ 2 } \frac { \pi l } { n } -
					\sin \frac { \pi k } { n } \sin^{ 2 } \frac { \pi l } { n }
					\right)^{ 2 }
					\nonumber
					\\
					&\overset{\eqref{eq: Probability Gambler Wins}}{=}
					P_{ G } ( k, l )
					\\
					&\text{and}
					\nonumber
					\\
					P_{ G } ( k, l^{ \ast } )
					&\overset{\eqref{eq: Probability Gambler Wins}}{=}
					\left(
					\cos \frac { \pi k } { n } \cos^{ 2 } \frac { \pi l^{ \ast } } { n } -
					\sin \frac { \pi k } { n } \sin^{ 2 } \frac { \pi l^{ \ast } } { n }
					\right)^{ 2 }
					\nonumber
					\\
					&=
					\left(
					\cos \frac { \pi k } { n } \cos^{ 2 } \left( \frac { 2 \pi l + \pi n} { 2 n } \right) -
					\sin \frac { \pi k } { n } \sin^{ 2 } \left( \frac { 2 \pi l + \pi n} { 2 n } \right)
					\right)^{ 2 }
					\nonumber
					\\
					&=
					\left(
					\cos \frac { \pi k } { n } \cos^{ 2 } \left( \frac { \pi l } { n } + \frac { \pi } { 2 } \right) -
					\sin \frac { \pi k } { n } \sin^{ 2 } \left( \frac { \pi l } { n } + \frac { \pi } { 2 } \right)
					\right)^{ 2 }
					\nonumber
					\\
					&=
					\left(
					\cos \frac { \pi k } { n } \sin^{ 2 } \frac { \pi l } { n } -
					\sin \frac { \pi k } { n } \cos^{ 2 } \frac { \pi l } { n }
					\right)^{ 2 }
					\nonumber
					\\
					&\overset{\eqref{eq: Probability Tosser Wins}}{=}
					P_{ T } ( k, l )
					\ .
				\end{align}
		\item	$l + \frac { n } { 2 } \geq n$. In this case, according to \eqref{eq: Explicit Dual l^*}, $l^{ \ast } = l - \frac { n } { 2 }$. Therefore,
				\begin{align}
					P_{ T } ( k, l^{ \ast } )
					&\overset{\eqref{eq: Probability Tosser Wins}}{=}
					\left(
					\cos \frac { \pi k } { n } \sin^{ 2 } \frac { \pi l^{ \ast } } { n } -
					\sin \frac { \pi k } { n } \cos^{ 2 } \frac { \pi l^{ \ast } } { n }
					\right)^{ 2 }
					\nonumber
					\\
					&=
					\left(
					\cos \frac { \pi k } { n } \sin^{ 2 } \left( \frac { 2 \pi l - \pi n} { 2 n } \right) -
					\sin \frac { \pi k } { n } \cos^{ 2 } \left( \frac { 2 \pi l - \pi n} { 2 n } \right)
					\right)^{ 2 }
					\nonumber
					\\
					&=
					\left(
					\cos \frac { \pi k } { n } \sin^{ 2 } \left( \frac { \pi l } { n } - \frac { \pi } { 2 } \right) -
					\sin \frac { \pi k } { n } \cos^{ 2 } \left( \frac { \pi l } { n } - \frac { \pi } { 2 } \right)
					\right)^{ 2 }
					\nonumber
					\\
					&=
					\left(
					\cos \frac { \pi k } { n } \cos^{ 2 } \frac { \pi l } { n } -
					\sin \frac { \pi k } { n } \sin^{ 2 } \frac { \pi l } { n }
					\right)^{ 2 }
					\nonumber
					\\
					&\overset{\eqref{eq: Probability Gambler Wins}}{=}
					P_{ G } ( k, l )
					\\
					&\text{and}
					\nonumber
					\\
					P_{ G } ( k, l^{ \ast } )
					&\overset{\eqref{eq: Probability Gambler Wins}}{=}
					\left(
					\cos \frac { \pi k } { n } \cos^{ 2 } \frac { \pi l^{ \ast } } { n } -
					\sin \frac { \pi k } { n } \sin^{ 2 } \frac { \pi l^{ \ast } } { n }
					\right)^{ 2 }
					\nonumber
					\\
					&=
					\left(
					\cos \frac { \pi k } { n } \cos^{ 2 } \left( \frac { 2 \pi l - \pi n} { 2 n } \right) -
					\sin \frac { \pi k } { n } \sin^{ 2 } \left( \frac { 2 \pi l - \pi n} { 2 n } \right)
					\right)^{ 2 }
					\nonumber
					\\
					&=
					\left(
					\cos \frac { \pi k } { n } \cos^{ 2 } \left( \frac { \pi l } { n } - \frac { \pi } { 2 } \right) -
					\sin \frac { \pi k } { n } \sin^{ 2 } \left( \frac { \pi l } { n } - \frac { \pi } { 2 } \right)
					\right)^{ 2 }
					\nonumber
					\\
					&=
					\left(
					\cos \frac { \pi k } { n } \sin^{ 2 } \frac { \pi l } { n } -
					\sin \frac { \pi k } { n } \cos^{ 2 } \frac { \pi l } { n }
					\right)^{ 2 }
					\nonumber
					\\
					&\overset{\eqref{eq: Probability Tosser Wins}}{=}
					P_{ T } ( k, l )
					\ .
				\end{align}
	\end{itemize}
	Ergo, for every rotation $r^{ k }$, $0 \leq k \leq n - 1$, chosen by the tosser during Phase $1$, both players end up with the same probability distribution.
\end{proof}

\begin{example} [Playing in the group $\langle \frac { 2 \pi } { 16 } \rangle$ continued] \label{xmp: Playing in the Group 2 pi By 16 Phase 2}
	This example continues and concludes Example \ref{xmp: Playing in the Group 2 pi By 16 Phase 1} by examining Phase $2$ (see Figure \ref{fig: The QC for the QTG}) of the QTG under the same assumption that the game is played in the cyclic rotation group $\langle \frac { 2 \pi } { 16 } \rangle$.

	Recall that Figure \ref{fig: Phase 1 Probabilities} visualizes the tosser's and the gambler's probabilities to win at the end of Phase $1$ as a function of the rotation $r^{ k }$, $0 \leq k \leq n - 1$, chosen by the tosser, and that the numerical values of the players' probabilities are contained in Table \ref{tbl: Phase 1 Probabilities}.

	Let us consider the scenario in which the tosser chose the rotation $r^{ 6 } = \frac { 12 \pi } { 16 }$ in Phase $1$. This choice was implemented via the use of the rotation gate $R_{ y } ( \frac { 12 \pi } { 16 } )$, which, in reality, causes a rotation through an angle $\frac { 6 \pi } { 16 }$ around the $y$ axis of the Bloch sphere, and gave player $1$ probability $0.854$ to win at the end of Phase $1$, as shown in Table \ref{tbl: Phase 1 Probabilities}. The motivation for such a choice could be that the tosser, being somewhat greedy, wanted a clear advantage over the gambler, while still allowing the latter a nonzero, but very slim, probability to win so as to keep appearances. However, player $2$ can still turn the tables and gain the upper hand by picking an appropriate measurement basis. This fact is visualized in Figure \ref{fig: Phase 2 Probabilities for k = 6}, which depicts both players' probabilities as a function of the measurement basis. The exact numerical values for each possible choice of player $2$ are given in Table \ref{tbl: Phase 2 Probabilities for k = 6}. With that in mind, it becomes evident that, if player $2$ picks $B_{ \frac { 16 \pi } { 12 } }$ as measurement basis, his probability to win becomes $0.854$, equal to the tosser's probability at the end of Phase $1$, which is another verification of Theorem \ref{thr: The QTG is Fair}.

	A final observation is that the tossers choice has affected the range of probabilities of both players. Although both have the same probability distribution as dictated by Theorem \ref{thr: The QTG is Fair}, the maximum probability for each player at the end of Phase $2$ is $0.854$, as determined by tosser's initial choice of $R_{ y } ( \frac { 12 \pi } { 16 } )$ during Phase $1$.
	\hfill $\triangleleft$
\end{example}

\begin{tcolorbox}
	[
		grow to left by = 1.00 cm,
		grow to right by = 1.00 cm,
		colback = MagentaVeryLight!03,			
		enhanced jigsaw,						
		sharp corners,
		toprule = 1.0 pt,
		bottomrule = 1.0 pt,
		leftrule = 0.1 pt,
		rightrule = 0.1 pt,
		sharp corners,
		center title,
		fonttitle = \bfseries
	]
	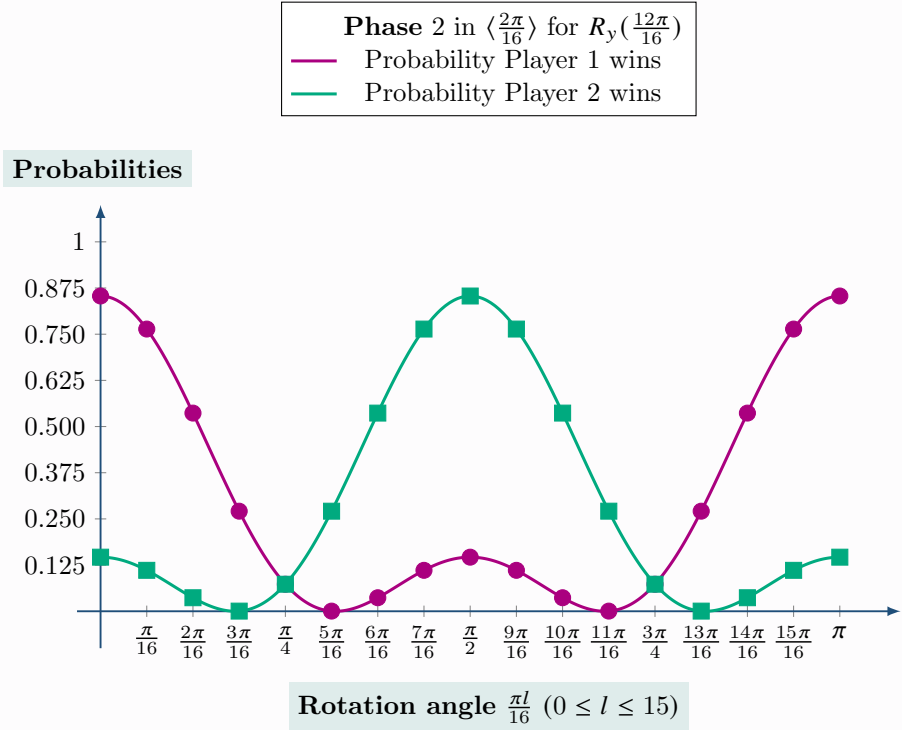
\begin{figure}[H]
		\centering
		\pgfplotsset { compat = 1.16 }
		\begin{tikzpicture} [ scale = 1.00 ]
			\def \n { 16 }
			\begin{axis}
				[
				width = 7.00 cm, height = 5.50 cm,
				axis line style = { WordBlueDarker, line width = 0.25 mm },
				axis x line = center,
				axis y line = middle,
				axis line style = { ->, > = { latex [ length = 7pt ] } },
				xmin = - 0.1, xmax = 3.40,
				ymin = - 0.1, ymax = 1.1,
				xtick =
				{
					pi * 1 / 16, pi * 2 / 16, pi * 3 / 16, pi * 4 / 16,
					pi * 5 / 16, pi * 6 / 16, pi * 7 / 16, pi * 8 / 16,
					pi * 9 / 16, pi * 10 / 16, pi * 11 / 16, pi * 12 / 16,
					pi * 13 / 16, pi * 14 / 16, pi * 15 / 16, pi * 16 / 16
				},
				xticklabels =
				{
					$\frac { \pi } { 16 }$, $\frac { 2 \pi } { 16 }$,
					$\frac { 3 \pi } { 16 }$, $\frac { \pi } { 4 }$,
					$\frac { 5 \pi } { 16 }$, $\frac { 6 \pi } { 16 }$,
					$\frac { 7 \pi } { 16 }$, $\frac { \pi } { 2 }$,
					$\frac { 9 \pi } { 16 }$, $\frac { 10 \pi } { 16 }$,
					$\frac { 11 \pi } { 16 }$, $\frac { 3 \pi } { 4 }$,
					$\frac { 13 \pi } { 16 }$, $\frac { 14 \pi } { 16 }$,
					$\frac { 15 \pi } { 16 }$, $\pi$
				},
				ytick = { 0, 0.125, 0.25, 0.375, 0.50, 0.625, 0.75, 0.875, 1.00 },
				yticklabels = { 0, 0.125, 0.250, 0.375, 0.500, 0.625, 0.750, 0.875, 1 },
				x label style =
				{
					at = { ( axis description cs: 0.5, -0.075 ) }, fill = WordVeryLightTeal, anchor = north
				},
				xlabel = { \textbf{Rotation angle} $\frac { \pi l } { 16 } \ (0 \leq l \leq 15)$ },
				y label style =
				{
					at = { ( axis description cs: 0.025, 1.040 ) }, fill = WordVeryLightTeal, anchor = south
				},
				ylabel = { \textbf{Probabilities} },
				x post scale = 2.0,
				y post scale = 1.5,
				grid = minor,
				trig format plots = rad,
				domain = 0 * pi : pi,
				samples = 800,
				declare function =
				{
					t6(\x) = ( - sin( 6 * pi / \n ) * cos( \x )^2 + cos( 6 * pi / \n ) * sin( \x )^2 )^2;
					g6(\x) = ( - sin( 6 * pi / \n ) * sin( \x )^2 + cos( 6 * pi / \n ) * cos( \x )^2 )^2;
				},
				legend style = { at = { ( 0.5, 1.2 ) }, anchor = south }
				]
				\addlegendimage { empty legend }
				\addlegendentry { \textbf{Phase} $2$ in $\langle \frac { 2 \pi } { 16 } \rangle$ for $R_{ y } ( \frac { 12 \pi } { 16 } )$ }
				\addplot [ RedPurple, line width = 0.4 mm ] { t6(\x) };
				\addlegendentry { Probability Player $1$ wins }
				\addplot [ GreenLighter2, line width = 0.4 mm ] { g6(\x) };
				\addlegendentry { Probability Player $2$ wins }
				\addplot [ color = RedPurple, draw = none, mark size = 3 pt, mark = * ] coordinates
				{
					( 0, { t6 ( 0 ) } )
					( pi * 1 / 16, { t6 ( pi * 1 / 16 ) } )
					( pi * 2 / 16, { t6 ( pi * 2 / 16 ) } )
					( pi * 3 / 16, { t6 ( pi * 3 / 16 ) } )
					( pi * 4 / 16, { t6 ( pi * 4 / 16 ) } )
					( pi * 5 / 16, { t6 ( pi * 5 / 16 ) } )
					( pi * 6 / 16, { t6 ( pi * 6 / 16 ) } )
					( pi * 7 / 16, { t6 ( pi * 7 / 16 ) } )
					( pi * 8 / 16, { t6 ( pi * 8 / 16 ) } )
					( pi * 9 / 16, { t6 ( pi * 9 / 16 ) } )
					( pi * 10 / 16, { t6 ( pi * 10 / 16 ) } )
					( pi * 11 / 16, { t6 ( pi * 11 / 16 ) } )
					( pi * 12 / 16, { t6 ( pi * 12 / 16 ) } )
					( pi * 13 / 16, { t6 ( pi * 13 / 16 ) } )
					( pi * 14 / 16, { t6 ( pi * 14 / 16 ) } )
					( pi * 15 / 16, { t6 ( pi * 15 / 16 ) } )
					( pi * 16 / 16, { t6 ( pi * 16 / 16 ) } )
				};
				\addplot [ color = GreenLighter2, draw = none, mark size = 3 pt, mark = square* ] coordinates
				{
					( 0, { g6 ( 0 ) } )
					( pi * 1 / 16, { g6 ( pi * 1 / 16 ) } )
					( pi * 2 / 16, { g6 ( pi * 2 / 16 ) } )
					( pi * 3 / 16, { g6 ( pi * 3 / 16 ) } )
					( pi * 4 / 16, { g6 ( pi * 4 / 16 ) } )
					( pi * 5 / 16, { g6 ( pi * 5 / 16 ) } )
					( pi * 6 / 16, { g6 ( pi * 6 / 16 ) } )
					( pi * 7 / 16, { g6 ( pi * 7 / 16 ) } )
					( pi * 8 / 16, { g6 ( pi * 8 / 16 ) } )
					( pi * 9 / 16, { g6 ( pi * 9 / 16 ) } )
					( pi * 10 / 16, { g6 ( pi * 10 / 16 ) } )
					( pi * 11 / 16, { g6 ( pi * 11 / 16 ) } )
					( pi * 12 / 16, { g6 ( pi * 12 / 16 ) } )
					( pi * 13 / 16, { g6 ( pi * 13 / 16 ) } )
					( pi * 14 / 16, { g6 ( pi * 14 / 16 ) } )
					( pi * 15 / 16, { g6 ( pi * 15 / 16 ) } )
					( pi * 16 / 16, { g6 ( pi * 16 / 16 ) } )
				};
			\end{axis}
		\end{tikzpicture}
		\caption{This figure shows the probabilities at the end of Phase $2$, assuming the QTG is played in $\langle \frac { 2 \pi } { 16 } \rangle$ and player $1$ has utilized the rotation gate $R_{ y } ( \frac { 12 \pi } { 16 } )$ during Phase $1$.}
		\label{fig: Phase 2 Probabilities for k = 6}
	\end{figure}
\end{tcolorbox}

\begin{tcolorbox}
	[
		grow to left by = 1.00 cm,
		grow to right by = 1.00 cm,
		colback = WordIceBlue!25,	%
		enhanced jigsaw,			
		sharp corners,
		boxrule = 0.1 pt,
		toprule = 0.1 pt,
		bottomrule = 0.1 pt
	]
	\begin{table}[H]
		\renewcommand{\arraystretch}{1.60}
		\caption{This table gives the probabilities for each player at the end of Phase $2$ when QTG is played in $\langle \frac { 2 \pi } { 16 } \rangle$ and player $1$ has utilized the rotation gate $R_{ y } ( \frac { 12 \pi } { 16 } )$ during Phase $1$.}
		\label{tbl: Phase 2 Probabilities for k = 6}
		\centering
		{\small
		\begin{tabular}
			{
				>{\centering\arraybackslash} m{1.50 cm} !{\vrule width 0.25 pt}
				>{\centering\arraybackslash} m{0.85 cm} !{\vrule width 0.25 pt}
				>{\centering\arraybackslash} m{0.85 cm} !{\vrule width 0.25 pt}
				>{\centering\arraybackslash} m{0.85 cm} !{\vrule width 0.25 pt}
				>{\centering\arraybackslash} m{0.85 cm} !{\vrule width 0.25 pt}
				>{\centering\arraybackslash} m{0.85 cm} !{\vrule width 0.25 pt}
				>{\centering\arraybackslash} m{0.85 cm} !{\vrule width 0.25 pt}
				>{\centering\arraybackslash} m{0.85 cm} !{\vrule width 0.25 pt}
				>{\centering\arraybackslash} m{0.85 cm} !{\vrule width 0.25 pt}
				>{\centering\arraybackslash} m{0.85 cm} !{\vrule width 0.25 pt}
			}
			\Xhline{4\arrayrulewidth}
			\multicolumn{10}{c}
			{ \cellcolor[HTML]{000000} { \color[HTML]{FFFFFF} Probabilities at the end of Phase $2$ in $\langle \frac { 2 \pi } { 16 } \rangle$ } }
			\\
			\Xhline{\arrayrulewidth}
			{ \cellcolor[HTML]{000000} { \color[HTML]{FFFFFF} $k$ } }
			&
			{ \cellcolor[HTML]{000000} { \color[HTML]{FFFFFF} $0$ } }
			&
			{ \cellcolor[HTML]{000000} { \color[HTML]{FFFFFF} $1, 15$ } }
			&
			{ \cellcolor[HTML]{000000} { \color[HTML]{FFFFFF} $2, 14$ } }
			&
			{ \cellcolor[HTML]{000000} { \color[HTML]{FFFFFF} $3, 13$ } }
			&
			{ \cellcolor[HTML]{000000} { \color[HTML]{FFFFFF} $4, 12$ } }
			&
			{ \cellcolor[HTML]{000000} { \color[HTML]{FFFFFF} $5, 11$ } }
			&
			{ \cellcolor[HTML]{000000} { \color[HTML]{FFFFFF} $6, 10$ } }
			&
			{ \cellcolor[HTML]{000000} { \color[HTML]{FFFFFF} $7, 9$ } }
			&
			{ \cellcolor[HTML]{000000} { \color[HTML]{FFFFFF} $8$ } }
			\\
			\Xhline{3\arrayrulewidth}
			{ \cellcolor[HTML]{000000} { \color[HTML]{FFFFFF} Tosser } }
			&
			$0.854$
			&
			$0.764$
			&
			$0.537$
			&
			$0.271$
			&
			$0.073$
			&
			$0$
			&
			$0.037$
			&
			$0.111$
			&
			$0.146$
			\\
			\Xhline{\arrayrulewidth}
			{ \cellcolor[HTML]{000000} { \color[HTML]{FFFFFF} Gambler } }
			&
			$0.146$
			&
			$0.111$
			&
			$0.037$
			&
			$0$
			&
			$0.073$
			&
			$0.271$
			&
			$0.537$
			&
			$0.764$
			&
			$0.854$
			\\
			\Xhline{4\arrayrulewidth}
		\end{tabular}
		}
		\renewcommand{\arraystretch}{1.0}
	\end{table}
\end{tcolorbox}

The last observation of Example \ref{xmp: Playing in the Group 2 pi By 16 Phase 2}, can be cast in the form of the next Corollary \ref{crl: Tosser's Choice Determines the Maximum Probability}.

\begin{corollary} [Tosser's choice determines the maximum probability] \label{crl: Tosser's Choice Determines the Maximum Probability}
	The tosser's choice of the rotation gate $R_{ y } ( \frac { 2 \pi k } { n } )$ during Phase $1$, determines the maximum probability at the end of Phase $2$ for both players.
\end{corollary}

\section{Discussion and conclusions} \label{sec: Discussion and Conclusions}

This paper introduces Quantum Tapsilou, a new quantum game inspired by the traditional Greek coin-tossing game known as tapsilou. While the classical version of the game features a straightforward and equal probability distribution, where both players have an equal $\frac { 1 } { 4 }$ probability of winning, Quantum Tapsilou introduces a new level of complexity by leveraging principles from quantum mechanics. In Quantum Tapsilou, both players still maintain an equal chance of winning, but these probabilities are now influenced by their prior choices and the utilization of quantum principles. The following two key innovations are at the core of Quantum Tapsilou.

Entanglement via rotation gates instead of relying on the commonly used Hadamard gates for creating entangled states with equal probability amplitudes. Quantum Tapsilou employs rotation gates, specifically $R_y$ rotation gates in order to generate Bell-like states with unequal probability amplitudes, introducing a layer of complexity into the game dynamics. Quantum Tapsilou makes integral use of groups in the sense that both players agree on a specific cyclic rotation group of order $n$, where $n$ should be a sufficiently large integer to provide players with additional choices and enhance the suspense of the game. This group serves as the foundation of the game and both players select rotations from this group to execute their actions using the corresponding $R_y$ rotation gates. The selection of this group and the specific rotations within it significantly affect the outcome, introducing a strategic element based on group theory.

By incorporating these quantum principles and group theory, Quantum Tapsilou preserves the main aspect of the classical tapsilou game while introducing complexity and variability into the probabilities of winning. Players must carefully consider their choices of rotations within the agreed-upon group and the entangled states generated by the $R_y$ rotation gates to maximize their chances of winning in this intriguing quantum adaptation of a traditional Greek game.
\bibliographystyle{ieeetr}
\bibliography{QTG}

\end{document}